\documentclass[runningheads,envcountsame,envcountsect]{llncs}

\usepackage[T1]{fontenc}
\usepackage{url}
\usepackage{xspace}
\usepackage{amssymb}
\usepackage{amsmath}
\usepackage{pifont} 
\usepackage{enumerate}
\usepackage{fp}
\usepackage{dsfont}
\usepackage{color}
\usepackage{comment}

\makeatletter
\def\alphspnewtheorem{\@ifstar{\@sthm}{\@alphSthm}}

\def\@alphthmcountersep{}
\def\@alphthmcounter#1{\noexpand\alph{#1}}

\def\@alphspnthm#1#2{%
  \@ifnextchar[{\@alphspxnthm{#1}{#2}}{\@alphspynthm{#1}{#2}}}
\def\@alphSthm#1{\@ifnextchar[{\@spothm{#1}}{\@alphspnthm{#1}}}

\def\@alphspxnthm#1#2[#3]#4#5{\expandafter\@ifdefinable\csname #1\endcsname
   {\@definecounter{#1}\@addtoreset{#1}{#3}%
   \expandafter\xdef\csname the#1\endcsname{\expandafter\noexpand
     \csname the#3\endcsname \noexpand\@alphthmcountersep \@alphthmcounter{#1}}%
   \expandafter\xdef\csname #1name\endcsname{#2}%
   \global\@namedef{#1}{\@spthm{#1}{\csname #1name\endcsname}{#4}{#5}}%
                              \global\@namedef{end#1}{\@endtheorem}}}

\def\@alphspynthm#1#2#3#4{\expandafter\@ifdefinable\csname #1\endcsname
   {\@definecounter{#1}%
   \expandafter\xdef\csname the#1\endcsname{\@alphthmcounter{#1}}%
   \expandafter\xdef\csname #1name\endcsname{#2}%
   \global\@namedef{#1}{\@spthm{#1}{\csname #1name\endcsname}{#3}{#4}}%
                               \global\@namedef{end#1}{\@endtheorem}}}
\makeatother

\alphspnewtheorem{alphtheorem}{Theorem}[theorem]{\bfseries}{\itshape}
\alphspnewtheorem{alphcase}[alphtheorem]{Case}{\itshape}{\rmfamily}
\alphspnewtheorem{alphconjecture}[alphtheorem]{Conjecture}{\itshape}{\rmfamily}
\alphspnewtheorem{alphcorollary}[alphtheorem]{Corollary}{\bfseries}{\itshape}
\alphspnewtheorem{alphdefinition}[alphtheorem]{Definition}{\bfseries}{\itshape}
\alphspnewtheorem{alphexample}[alphtheorem]{Example}{\itshape}{\rmfamily}
\alphspnewtheorem{alphexercise}[alphtheorem]{Exercise}{\itshape}{\rmfamily}
\alphspnewtheorem{alphlemma}[alphtheorem]{Lemma}{\bfseries}{\itshape}
\alphspnewtheorem{alphnote}[alphtheorem]{Note}{\itshape}{\rmfamily}
\alphspnewtheorem{alphproblem}[alphtheorem]{Problem}{\itshape}{\rmfamily}
\alphspnewtheorem{alphproperty}[alphtheorem]{Property}{\itshape}{\rmfamily}
\alphspnewtheorem{alphproposition}[alphtheorem]{Proposition}{\bfseries}{\itshape}
\alphspnewtheorem{alphquestion}[alphtheorem]{Question}{\itshape}{\rmfamily}
\alphspnewtheorem{alphsolution}[alphtheorem]{Solution}{\itshape}{\rmfamily}
\alphspnewtheorem{alphremark}[alphtheorem]{Remark}{\itshape}{\rmfamily}

\newcommand{\textcent}{\rlap/{\textrm{c}}}

\newcommand{\COMMENT}[3]{
{\sf
\begin{description}
\item[#1 #2] #3 
\end{description}\noindent
}}

\renewcommand{\COMMENT}[3]{\ignorespaces}

\newcommand{\IR}[1]{\text{\normalfont\sffamily\small#1}} 
\def\void{}
\newcommand{\infrule}[3][\void]{%
  {\renewcommand\arraystretch{1.25}
    \ifx\void#1\else#1\hspace{0.5em}\fi
    \begin{array}[c]{@{\hspace*{1em}}c@{\hspace*{1em}}}#2\\\hline #3
    \end{array}}}

\newcommand{\op}[1]{\mathsf{{#1}}}

\newcommand{\dom}{\mathrm{dom}}
\newcommand{\vars}{\mathrm{vars}}

\newcommand{\sgi}{\mathrm{sgi}}
\newcommand{\smgi}{\mathrm{smgi}}

\newcommand{\nojunk}{\mathrm{nojunk}}

\newcommand{\calD}{\mathcal{D}}

\newcommand{\pre}{\mathrm{pre}}
\newcommand{\back}{\mathrm{B}}
\newcommand{\fore}{\mathrm{F}}

\newcommand{\Base}{{\ensuremath{\mathrm{Base}}}}

\newcommand{\beagle}{\textit{Beagle}\xspace}

\let\oldvdash\vdash
\def\vdash{\:\oldvdash\:}
\let\oldcup\cup
\def\cup{\,\oldcup\,}
\let\oldcap\cap
\def\cap{\:\oldcap\:}

\newcommand{\cX}{\mathcal{X}}

\newcommand{\GndTh}{\mathrm{GndTh}}
\newcommand{\abstr}{\mathrm{abstr}}
\newcommand{\unabstr}{\mathrm{unabstr}}

\newcommand{\SP}{\mathit{SP}}
\newcommand{\SSP}{\mathrm{SSP}}
\newcommand{\HSP}{\mathrm{HSP}}

\newcommand{\TSigma}{\mathrm{T}_\Sigma}

\newcommand{\BAlgs}{\mathcal{B}}
\newcommand{\JAlgs}{\mathcal{J}}
\newcommand{\rstr}[1]{{|_{#1}}}

\newcommand{\gtc}{\succ^\mathrm{c}}

\newcommand{\Inf}{\mathcal{I}}
\newcommand{\Red}{\mathcal{R}}
\newcommand{\RInf}{\Red_\mathrm{Inf}}
\newcommand{\RCl}{\Red_\mathrm{Cl}}
\newcommand{\RedFlt}{\Red^\mathcal{S}}
\newcommand{\RInfFlt}{\Red_\mathrm{Inf}^\mathcal{S}}
\newcommand{\RClFlt}{\Red_\mathrm{Cl}^\mathcal{S}}
\newcommand{\RedSgi}{\Red^\mathcal{H}}
\newcommand{\RInfSgi}{\Red_\mathrm{Inf}^\mathcal{H}}
\newcommand{\RClSgi}{\Red_\mathrm{Cl}^\mathcal{H}}
\newcommand{\concl}{\mathrm{concl}}

\newcommand{\EA}{\mathrm{E}_I}
\newcommand{\DA}{\mathrm{D}_I}
\newcommand{\EADA}{\EA \cup \DA}

\newcommand{\qbar}{\bar{q}}
\newcommand{\ie}{i.\,e.}
\newcommand{\eg}{e.\,g.}
\newcommand{\wrt}{w.\,r.\,t.}

\begin{document}

  \title{Hierarchic Superposition Revisited\thanks{The final authenticated version is
      available online at \email{https://doi.org/TBD}.}}
\author{Peter Baumgartner\inst{1}
  \\
  \and
  Uwe Waldmann\inst{2}
}
\authorrunning{P. Baumgartner and U. Waldmann}

\institute{Data61/CSIRO, ANU Computer Science and Information Technology (CSIT),\\ Acton, Australia,
\email{Peter.Baumgartner@data61.csiro.au}
  \and
Max-Planck-Institut f\"ur Informatik, Saarbr\"ucken, Germany,\\
\email{uwe@mpi-inf.mpg.de}}

\maketitle

\begin{abstract}
Many applications of automated deduction require reasoning in
first-order logic modulo background theories, in particular some form of
integer arithmetic. A major unsolved research challenge is to design
theorem provers that are ``reasonably complete'' even in the presence of
free function symbols ranging into a background theory sort. The
hierarchic superposition calculus of Bachmair, Ganzinger, and Waldmann
already supports such symbols, but, as we demonstrate, not optimally.
This paper aims to rectify the situation by introducing a novel form of
clause abstraction, a core component in the hierarchic superposition
calculus for transforming clauses into a form needed for internal
operation. We argue for the benefits of the resulting calculus and
provide two new completeness results: one for the fragment where all
background-sorted terms are ground and another one for a special
case of linear (integer or rational) arithmetic as a background theory.
\end{abstract}

\section{Introduction}
\label{sec:introduction}
Many applications of automated deduction require reasoning
with respect to a combination of a background theory,
say integer arithmetic, and a foreground theory that extends
the background theory by new sorts such as $\mathit{list}$,
new operators, such as
$\mathit{cons} : \mathit{int} \times \mathit{list} \to \mathit{list}$
and $\mathit{length} : \mathit{list} \to \mathit{int}$,
and first-order axioms.
Developing corresponding
automated reasoning systems that are also able to deal with quantified formulas has
recently been an active area of research. One major line of research is concerned
with extending (SMT-based) solvers~\cite{NieOT-JACM-06} for the quantifier-free case
by instantiation heuristics for
quantifiers~\cite[\eg]{GeBT-CADE-07,Ge:DeMoura:CompleteInstantiation:CAV:2009}.
Another line of research is concerned with adding black-box reasoners for specific
background theories to first-order automated reasoning methods
(resolution~\cite{Bachmair:Ganzinger:Waldmann:TheoremProvingHierarchical:AAECC:94,Korovin:Voronkov:IntegratingLIASuperposition:CSL:2007,Althaus:EtAt:SuperpositionLA:FroCos:2009},
sequent calculi~\cite{Ruemmer:SequentCalculusLIA:LPAR:2008}, instantiation
methods~\cite{Ganzinger:Korovin:ThInst:LPAR:2006,Baumgartner:Fuchs:Tinelli:MELIA:LPAR:2008,Baumgartner:Tinelli:MEET:CADE:2011},
etc). In both cases, a major unsolved research challenge is to provide reasoning
support that is ``reasonably complete'' in practice, so that the systems can be used
more reliably for both proving theorems and finding counterexamples.

In~\cite{Bachmair:Ganzinger:Waldmann:TheoremProvingHierarchical:AAECC:94},
Bachmair, Ganzinger, and Waldmann introduced the hierarchical
superposition calculus as a generalization of the superposition
calculus for black-box style theory reasoning.
Their calculus works in a framework of hierarchic specifications.
It tries to prove the
unsatisfiability of a set of clauses with respect to interpretations
that extend a background model such as the integers with linear arithmetic
conservatively, that is, without
identifying distinct elements of old sorts (``confusion'') and without
adding new elements to old sorts (``junk'').
While confusion can be detected by first-order theorem proving
techniques, junk can not -- in fact, the set of logical consequences of
a hierarchic specifications is usually not recursively enumerable.
Refutational completeness can therefore only be guaranteed if one
restricts oneself to sets of formulas where junk can be excluded a priori.
The property introduced
by Bachmair, Ganzinger, and Waldmann
for this purpose
is called ``sufficient completeness with respect to simple instances''.
Given this property, their calculus is refutationally
complete for clause sets that are fully abstracted (\ie, where
no literal contains both foreground and background symbols).
Unfortunately their full abstraction rule may 
destroy
sufficient completeness with respect to simple instances.
We show that this problem can be avoided by using a new form of clause
abstraction and a suitably modified hierarchical superposition calculus.
Since the new hierarchical superposition calculus
is still
refutationally complete and the new abstraction rule
is guaranteed to preserve
sufficient completeness with respect to simple instances,
the new combination is strictly more powerful than the~old~one.

In practice, sufficient completeness is a rather restrictive property.
While there are application areas where
one knows in advance that every input is sufficiently complete,
in most cases this does not hold.
As a user of an automated theorem prover,
one would like to see a best effort behavior:
The prover might for instance try to \emph{make} the input
sufficiently complete by adding further theory axioms.
In the calculus by Bachmair, Ganzinger, and Waldmann, however,
this does not help at all:
The restriction to a particular kind of instantiations
(``simple instances'') renders theory axioms essentially unusable
in refutations.
We show that this can be prevented
by introducing
two kinds of variables of the background theory instead of one,
that can be instantiated in different ways,
making our calculus significantly ``more complete'' in practice.
We also include a definition rule in the calculus that can be
used to establish sufficient completeness by linking
foreground terms to background parameters,
\looseness=-1
thus allowing the background prover to reason about these terms.

The following trivial example demonstrates the problem.
Consider
the clause set $N = \{ C\}$ where $C = \op f(1) < \op f(1)$. Assume that the background
theory is integer 
arithmetic and that $\op f$ is an integer-sorted operator from the foreground (free)
signature.
Intuitively, one would expect $N$ to be unsatisfiable.
However, $N$
is not sufficiently complete, and
it admits models in which $\op f(1)$ is
interpreted as some junk element $\textcent$, an element of the domain of the integer sort
that is not a numeric constant. So both the calculus
in~\cite{Bachmair:Ganzinger:Waldmann:TheoremProvingHierarchical:AAECC:94} and ours
are excused to not find a refutation. 
To fix that, one could add an instance  $C' = \neg(\op f(1) < \op f(1))$  of the
irreflexivity axiom $\neg(x < x)$. The resulting set $N' = \{ C,\ C' \}$ is (trivially)
sufficiently complete as it has no models at all. However, the calculus
in~\cite{Bachmair:Ganzinger:Waldmann:TheoremProvingHierarchical:AAECC:94} is not
helped by adding $C'$,
since the abstracted version of $N'$ is again
not sufficiently complete and admits a model that interprets $\op f(1)$ as
$\textcent$. Our abstraction mechanism always preserves sufficient completeness and
our calculus will find a refutation.

With this example one could think
that replacing the abstraction mechanism
in~\cite{Bachmair:Ganzinger:Waldmann:TheoremProvingHierarchical:AAECC:94} 
with ours gives all the advantages of our calculus. But this is not the case. 
Let $N'' = \{ C,\ \neg(x < x) \}$ be obtained by adding the more realistic axiom 
$\neg(x < x)$. 
The set  $N''$ is still sufficiently complete with our approach thanks to having two
kinds of variables at disposal, but it is not sufficiently complete in the sense
of~\cite{Bachmair:Ganzinger:Waldmann:TheoremProvingHierarchical:AAECC:94}. 
Indeed, in that calculus
adding background theory axioms \emph{never} helps
to gain sufficient completeness, as variables there have only one kind.

Another alternative to make $N$ sufficiently complete is by adding a clause that forces
$\op f(1)$ to be equal to some background domain element. For instance, one can add 
a ``definition'' for $\op f(1)$, that is, a clause $\op f(1) \approx \alpha$, where $\alpha$ is a fresh symbolic
constant belonging to the background signature (a ``parameter''). The set 
$N''' = \{ C, \op f(1) \approx \alpha \}$ is sufficiently complete and it admits refutations 
with both calculi. The definition rule in our calculus mentioned above will generate this
definition automatically. Moreover, the set $N$ belongs to a syntactic fragment for 
which we can guarantee not only sufficient completeness (by means of the definition
rule) but also refutational completeness.

We present the new calculus in detail and provide
a general completeness result, modulo compactness of the background theory,
and two specific completeness results for 
clause sets that do not require compactness --
one for the fragment where all
background-sorted terms are ground and another one for a special
case of linear (integer or rational) arithmetic as a background theory.

We also report on
experiments with a
prototypical implementation on the TPTP problem library~\cite{Sutcliffe:TPTP:2017}. 

Sections~\ref{sec:introduction}--\ref{sec:refutational-completeness},
\ref{sec:define}--\ref{sec:gbt}, and
\ref{sec:experiments} of this paper are a substantially expanded and revised
version of~\cite{DBLP:conf/cade/BaumgartnerW13}.
A preliminary version of Sect.~\ref{sec:linear-arithmetic}
has appeared in~\cite{BaumgartnerWaldmann2013MACIS}.%
Lemmas that had to be left out in the final authenticated version
by lack of space are marked with a letter following the number
(\eg, Lemma~\ref{lemma:target-stable-under-unabstraction}).

\paragraph{Related Work.} The relation with the predecessor calculus
in~\cite{Bachmair:Ganzinger:Waldmann:TheoremProvingHierarchical:AAECC:94} is
discussed above and also further below. What we say there also applies
to other developments rooted in that
calculus,~\cite[\eg]{Althaus:EtAt:SuperpositionLA:FroCos:2009}.  The specialized
version of hierarchic superposition in~\cite{Kruglov:Weidenbach:MACIS:2012} will be
discussed in Sect.~\ref{sec:define} below. 
The resolution calculus
in~\cite{Korovin:Voronkov:IntegratingLIASuperposition:CSL:2007} has built-in
inference rules for linear (rational) arithmetic, but is complete only under
restrictions that effectively prevent quantification over rationals.
Earlier work on integrating theory reasoning into model
evolution~\cite{Baumgartner:Fuchs:Tinelli:MELIA:LPAR:2008,Baumgartner:Tinelli:MEET:CADE:2011}
lacks the treatment of background-sorted foreground function symbols.
The same applies to the sequent calculus
in~\cite{Ruemmer:SequentCalculusLIA:LPAR:2008}, which treats linear arithmetic with
built-in rules for quantifier elimination. 
The instantiation method in~\cite{Ganzinger:Korovin:ThInst:LPAR:2006} requires an
answer-complete solver for the background theory to enumerate concrete 
solutions of background constraints, not just a decision procedure. 
All these approaches have in common that they integrate specialized reasoning for
background theories into a general first-order reasoning method. A conceptually
different approach consists in using first-order theorem provers as (semi-)decision
procedures for specific
theories in DPLL(T)(-like) architectures~\cite{Moura:Bjoerner:EngineeringDPLLTSaturation:IJCAR:2008,Armando:etal:NewResultsSatisfiability:ACMTCL:2009,Bonacina:etal:speculative-inferences:JAR:2011}. 
Notice that in this context the theorem provers do not need to reason
modulo background theories themselves, and indeed they don't.
The calculus and system in~\cite{Moura:Bjoerner:EngineeringDPLLTSaturation:IJCAR:2008}, for
instance, integrates superposition and DPLL(T). From DPLL(T) it inherits splitting of ground non-unit clauses
into their unit components, which determines a (backtrackable) model candidate
$M$. The superposition 
inference rules are applied to elements from $M$ and a current clause set $F$. 
The superposition component guarantees refutational completeness for pure first-order
clause logic. Beyond that, for clauses containing
background-sorted variables, (heuristic) instantiation is needed. 
Instantiation is done with ground terms that are provably equal 
\wrt\ the equations in $M$ to some ground term in $M$ in order to advance the derivation.
The limits of that method can
be illustrated with an (artificial but simple) example.  Consider the unsatisfiable clause
set $\{ i \leq  j \lor  \op P(i+1, x) \lor  \op P(j+2, x),$
$i \leq  j \lor  \lnot \op P(i+3, x) \lor  \lnot \op P(j+4,x) \}$
where 
$i$ and $j$ are integer-sorted variables and $x$ is a foreground-sorted variable. 
Neither splitting into unit clauses, superposition calculus rules, nor instantiation applies, and so
the derivation gets stuck with an inconclusive result. By contrast, the clause set
belongs to a fragment that entails sufficient completeness (``no background-sorted
foreground function symbols'') and hence is refutable by our calculus. On the other
hand, heuristic instantiation does have a place in our calculus, but we leave that
for future work.

\section{Signatures, Clauses, and Interpretations}
\label{sec:sig-clause-interp}
We work in the context of standard many-sorted 
logic with first-order signatures comprised of sorts and operator
(or function) symbols of given arities over
these sorts.
A \emph{signature} is a pair $\Sigma = (\Xi,\Omega)$,
where $\Xi$ is a set of \emph{sorts} and
$\Omega$ is a set of \emph{operator symbols} over $\Xi$.
If $\cX$ is a set of sorted variables with sorts in $\Xi$, then
the set of well-sorted terms over $\Sigma = (\Xi,\Omega)$ and $\cX$
is denoted by $\TSigma(\cX)$;
$\TSigma$ is short for $\TSigma(\emptyset)$.
We require that $\Sigma$ is a \emph{sensible} signature,
\ie, that $\TSigma$ has no empty sorts.
As usual, we write $t[u]$ to indicate that the term $u$ is a
(not necessarily proper) subterm of the term $t$.
The position of $u$ in $t$ is left implicit. 

A \emph{$\Sigma$-equation} is an unordered pair $(s,t)$,
usually written $s \approx t$,
where $s$ and $t$ are terms from $\TSigma(\cX)$
of the same sort.
For simplicity, we use equality as the only predicate in our language.
Other predicates can always be
encoded as a function into a set with one
distinguished element, so that
a non-equational atom is turned into an
equation $P(t_1,\ldots,t_n) \approx {\sl true}_P^{}$;
this is usually abbreviated by $P(t_1,\ldots,t_n)$.\footnote{%
   Without loss of generality we assume that there exists a
   distinct sort for every predicate.}
A \emph{literal} is an equation $s \approx t$ or a negated equation
$\lnot(s \approx t)$, also written as $s \not\approx t$.
A \emph{clause} is a multiset of literals, usually written as a
disjunction;
the empty clause, denoted by $\Box$ is a contradiction.
If $F$ is a term, equation, literal or clause,
we denote by \emph{$\vars(F)$} the set
of variables that occur in $F$.
We say $F$ is \emph{ground} if $\vars(F) = \emptyset$ 

A \emph{substitution $\sigma$} is a mapping 
from variables to terms that is sort respecting,
that is, maps each variable $x \in \cX$ to a term of the same sort.
Substitutions are homomorphically extended to terms as usual.
We write substitution application in postfix form.
A term $s$ is an \emph{instance} of a term $t$ if there is a
substitution $\sigma$ such that $t\sigma = s$.
All these
notions carry over to equations, literals and clauses in the obvious way.
The composition $\sigma\tau$ of the substitutions $\sigma$ and $\tau$
is the substitution that maps every variable $x$ to~$(x\sigma)\tau$.

The \emph{domain} of a substitution $\sigma$ is 
the set $\dom(\sigma) = \{ x \mid x \neq x\sigma \}$.
We use only substitutions with finite domains,
written as $\sigma = [x_1\mapsto t_1,\ldots,x_n\mapsto t_n]$ where $\dom(\sigma) = \{x_1,\ldots,x_n\}$.
A \emph{ground substitution} is a substitution that maps every variable
in its
domain to a ground term.
A \emph{ground instance} of $F$ is obtained by
applying some ground substitution with domain (at least) $\vars(F)$ to it.

A \emph{$\Sigma$-interpretation} $I$ consists of a $\Xi$-sorted
family of carrier sets $\{I_\xi\}_{\xi \in \Xi}$
and of a function
$I_f : I_{\xi_1} \times \cdots \times I_{\xi_n} \to I_{\xi_0}$
for every $f : \xi_1 \ldots \xi_n \to \xi_0$ in $\Omega$.
The \emph{interpretation} $t^I$ of a ground term $t$ is defined recursively
by $f(t_1,\ldots,t_n)^I = I_f(t_1^I,\ldots,t_n^I)$ for $n \geq 0$.
An interpretation $I$ is called \emph{term-generated}, if
every element of an $I_\xi$ is the interpretation of some
ground term of sort $\xi$.
An interpretation $I$ is said to \emph{satisfy} a
ground equation $s \approx t$,
if $s$ and $t$ have the same interpretation
in $I$;
it is said to \emph{satisfy} a negated
ground equation $s \not\approx t$, if $s$ and $t$ do not have
the same interpretation
in $I$.
The interpretation $I$ \emph{satisfies} a
ground clause $C$ if
at least one of the literals of $C$ is satisfied by $I$.
We also say that a ground clause $C$ is \emph{true in $I$},
if $I$ satisfies $C$;
and that $C$ is \emph{false in $I$}, otherwise.
A term-generated interpretation $I$ is said to \emph{satisfy} a
non-ground clause $C$ if
it satisfies all ground instances $C\sigma$;
it is called a \emph{model} of a set $N$ of clauses,
if it satisfies all clauses of $N$.\footnote{%
   This restriction to term-generated interpretations as models is possible
   since we are only concerned with refutational theorem proving,
   \ie, with the derivation of a contradiction.}
We abbreviate the fact that $I$ is a model of $N$
by $I \models N$; $I \models C$ is short for $I \models \{C\}$.
We say that $N$ \emph{entails} $N'$, and write $N \models N'$,
if every model of $N$ is a model of $N'$;
$N \models C$ is short for $N \models \{C\}$.
We say that $N$ and $N'$ are \emph{equivalent},
if $N \models N'$ and $N' \models N$.

If $\JAlgs$ is a class of $\Sigma$-interpretations,
a $\Sigma$-clause or clause set is called \emph{$\JAlgs$-satisfiable}
if at least one $I \in \JAlgs$ satisfies the clause or clause set;
otherwise it is called \emph{$\JAlgs$-unsatisfiable}.

A \emph{specification} is a pair $\SP = (\Sigma,\JAlgs)$,
where $\Sigma$ is a signature and
$\JAlgs$ is a class of term-generated $\Sigma$-interpretations
called \emph{models}
of the specification $\SP$.
We assume that $\JAlgs$ is closed under isomorphisms.

We say that a class of $\Sigma$-interpretations $\JAlgs$
or a specification $(\Sigma,\JAlgs)$
is \emph{compact},
if every infinite set of $\Sigma$-clauses
that is $\JAlgs$-unsatisfiable
has a finite subset that is also $\JAlgs$-unsatisfiable.

\section{Hierarchic Theorem Proving}
\label{sec:hierarchic-tp}
In hierarchic theorem proving, we consider a scenario in which
a general-purpose foreground theorem prover
and a specialized background prover
cooperate to derive a contradiction from a set of clauses.
In the sequel, we will usually abbreviate ``foreground'' and
``background'' by ``FG'' and ``BG''.

The BG prover accepts as input sets of clauses over
a \emph{BG signature} $\Sigma_\back = (\Xi_\back,\Omega_\back)$.
Elements of $\Xi_\back$ and $\Omega_\back$ are called
\emph{BG sorts} and \emph{BG operators},
respectively.
We fix an infinite set $\cX_\back$ of \emph{BG variables}
of sorts in $\Xi_\back$.
Every BG variable has (is labeled with) a \emph{kind},
which is either \emph{``abstraction''} or \emph{``ordinary''}.
Terms over $\Sigma_\back$ and $\cX_\back$
are called \emph{BG terms}.
A BG term is called \emph{pure}, if it does not contain
ordinary variables; otherwise it is \emph{impure}. These notions 
apply analogously to equations, literals, clauses, and clause sets.

The BG prover decides the satisfiability of
$\Sigma_\back$-clause sets with respect to a
\emph{BG specification} $(\Sigma_\back,\BAlgs)$,
where $\BAlgs$ is a class of term-generated $\Sigma_\back$-interpretations
called \emph{BG models}.
We assume that $\BAlgs$ is closed under isomorphisms.

In most applications of hierarchic theorem proving,
the set of BG operators $\Omega_\back$ contains
a set of distinguished constant symbols $\Omega_\back^D \subseteq \Omega_\back$
that has the property that $d_1^I \neq d_2^I$ for any two distinct
$d_1, d_2 \in \Omega_\back^D$ and every BG model $I \in \BAlgs$.
We refer to these constant symbols as
\emph{(BG) domain elements}.

While we permit arbitrary classes of BG models,
in practice the following three cases are most relevant:
\begin{enumerate}[(1)]
\item
  $\BAlgs$ consists of exactly one $\Sigma_\back$-interpretation
  (up to isomorphism),
  say, the integer numbers over a signature containing
  all integer constants as domain elements
  and ${\leq}, {<}, +, -$ 
  with the expected arities.
  In this case, $\BAlgs$ is trivially compact; in fact, a set $N$ of
  $\Sigma_\back$-clauses is $\BAlgs$-unsatisfiable
  if and only if some clause of $N$ is $\BAlgs$-unsatisfiable.
\item
  $\Sigma_\back$ is extended by an infinite number of
  \emph{parameters}, that is, additional constant
  symbols. While all interpretations in $\BAlgs$ share the same
  carrier sets $\{I_\xi\}_{\xi \in \Xi_\back}$ and interpretations of
  non-parameter symbols,
  parameters may be interpreted freely by arbitrary elements
  of the appropriate $I_\xi$.
  The class $\BAlgs$ obtained in this way is in general not compact;
  for instance the infinite set of clauses
  $\{\,n \leq \beta \mid n \in \mathbb{N}\,\}$,
  where $\beta$ is a parameter, is unsatisfiable
  in the integers, but every finite subset is satisfiable.
\item
  $\Sigma_\back$ is again extended by parameters,
  however, $\BAlgs$ is now the class of all
  interpretations that satisfy some first-order theory,
  say, the first-order theory of linear integer arithmetic.\footnote{%
    To satisfy the technical requirement that all interpretations in $\BAlgs$
    are term-generated, we assume that in this case $\Sigma_\back$
    is suitably extended by an infinite set of constants (or by one
    constant and one unary function symbol) that are not used
    in any input formula or theory axiom.}
  Since $\BAlgs$ corresponds to a first-order theory, compactness is
  recovered. It should be noted, however, that
  $\BAlgs$ contains non-standard models,
  so that for instance
  the clause set $\{\,n \leq \beta \mid n \in \mathbb{N}\,\}$
  is now satisfiable (\eg, $\mathbb{Q} \times \mathbb{Z}$ with a
  lexicographic ordering is a model).
\end{enumerate}

The FG theorem prover accepts as inputs clauses over
a signature $\Sigma = (\Xi,\Omega)$,
where $\Xi_\back \subseteq \Xi$ and
$\Omega_\back \subseteq \Omega$.
The sorts in $\Xi_\fore = \Xi \setminus \Xi_\back$ and
the operator symbols in $\Omega_\fore = \Omega \setminus \Omega_\back$
are called \emph{FG sorts} and \emph{FG operators}.
Again we fix an infinite set $\cX_\fore$ of \emph{FG variables}
of sorts in $\Xi_\fore$.
All FG variables have the kind ``ordinary''.
We define $\cX = \cX_\back \cup \cX_\fore$.

In examples we will use $\{0, 1, 2, \dots\}$ to denote BG domain elements,
$\{{+}, {-}, {<},$ ${\leq}\}$
to denote (non-parameter) BG operators,
and the possibly subscripted letters $\{x,y\}$,
$\{X,Y\}$, $\{\alpha,\beta\}$, and $\{\op a, \op b, \op c, \op f, \op g\}$
to denote ordinary variables, abstraction
variables, parameters, and FG operators,
respectively. The letter $\zeta$ denotes an ordinary variable or an abstraction variable.

We call a term in $\TSigma(\cX)$
a \emph{FG term}, if it is not a BG term,
that is, if it contains at least one FG operator or
FG variable
(and analogously for literals or clauses).
We emphasize that for a FG operator
$\op f : \xi_1 \ldots \xi_n \to \xi_0$ in~$\Omega_\fore$
any of the $\xi_i$ may be a BG sort,
and that consequently
FG terms may have BG sorts.

If $I$ is a $\Sigma$-interpretation, the \emph{restriction} of $I$
to $\Sigma_\back$,
written $I\rstr{\Sigma_\back}$, is the $\Sigma_\back$-interpretation
that is obtained from
$I$ by removing all carrier sets
$I_\xi$ for $\xi \in \Xi_\fore$
and all functions
$I_f$ for $f \in \Omega_\fore$.
Note that $I\rstr{\Sigma_\back}$ is not necessarily term-generated
even if $I$ is term-generated.
In hierarchic theorem proving,
we are only interested in $\Sigma$-interpretations
that extend some model in $\BAlgs$
and neither collapse any of its sorts nor add new elements to them,
that is, in $\Sigma$-interpretations $I$ for which
$I\rstr{\Sigma_\back} \in \BAlgs$.
We call such a $\Sigma$-interpretation a \emph{$\BAlgs$-interpretation}.

Let $N$ and $N'$ be two sets of $\Sigma$-clauses.
We say that $N$ \emph{entails} $N'$ \emph{relative to} $\BAlgs$
(and write $N \models_\BAlgs N'$),
if every model of $N$ whose restriction to $\Sigma_\back$ is in $\BAlgs$
is a model of $N'$.
Note that $N \models_\BAlgs N'$ follows from $N \models N'$.
If $N \models_\BAlgs \Box$, we call $N$
\emph{$\BAlgs$-unsatisfiable};
otherwise, we call it \emph{$\BAlgs$-satisfiable}.\footnote{%
  If $\Sigma = \Sigma_\back$,
  this definition coincides with
  the definition of satisfiability \wrt~a class of interpretations
  that was given in Sect.~\ref{sec:sig-clause-interp}.
  A set $N$ of BG clauses is $\BAlgs$-satisfiable
  if and only if some interpretation of $\BAlgs$ is a model of $N$.
}

Our goal in refutational hierarchic theorem proving
is to check
whether a given set of $\Sigma$-clauses $N$
is false in all $\BAlgs$-interpretations,
or equivalently,
whether $N$ is $\BAlgs$-unsatisfiable.

\bigskip

We say that a substitution $\sigma$ is \emph{simple}
if
$X\sigma$ is a pure BG term
for every abstraction
variable $X \in \dom(\sigma)$.
For example,  $[x \mapsto 1+Y+\alpha]$, $[X \mapsto 1+Y+\alpha]$ and
$[x \mapsto \op f(1)]$ all are simple,
whereas $[X \mapsto 1+y+\alpha]$ and $[X \mapsto \op f(1)]$ are not.
Let $F$ be a clause or (possibly infinite) clause set.
By \emph{$\sgi(F)$} we denote the set of simple ground instances of $F$,
that is,
the set of all ground instances of (all clauses in) $F$ obtained by simple
ground substitutions. 

\bigskip

For a BG specification $(\Sigma_\back,\BAlgs)$,
we define $\GndTh(\BAlgs)$ as the set of all ground $\Sigma_\back$-formulas
that are satisfied by every $I \in \BAlgs$.

\begin{definition}[Sufficient completeness]
\label{def:sufficient-completeness}
A $\Sigma$-clause set $N$ is called \emph{sufficiently complete \wrt\ simple instances}
if for every $\Sigma$-model $J$ of $\sgi(N) \cup \GndTh(\BAlgs)$\footnote{%
 In contrast to~\cite{Bachmair:Ganzinger:Waldmann:TheoremProvingHierarchical:AAECC:94},
 we include $\GndTh(\BAlgs)$ in the definition of sufficient completeness.
 (This is independent of the abstraction method;
 it would also have been useful
 in~\cite{Bachmair:Ganzinger:Waldmann:TheoremProvingHierarchical:AAECC:94}.)
}
and every ground BG-sorted
FG term $s$ there is
a ground BG term $t$ such that $J \models s \approx t$.\footnote{%
  Note that $J$ need \emph{not} be a $\BAlgs$-interpretation.
} \qed
\end{definition}
For brevity, we will from now on omit the phrase ``\wrt~simple instances''
and speak only of ``sufficient completeness''.
It should be noted, though, that our definition differs from
the classical definition of sufficient completeness in the
literature on algebraic specifications.

\section{Orderings}
\label{sec:orderings}
A \emph{hierarchic reduction ordering}
is a strict, well-founded ordering on
terms that is compatible with contexts,
\ie, $s \succ t$ implies $u[s] \succ u[t]$, and
stable under simple substitutions,
\ie, $s \succ t$ implies $s\sigma \succ t\sigma$
for every simple $\sigma$.
In the rest of this paper we assume such a hierarchic reduction
ordering $\succ$ that satisfies all of the following:
(i) $\succ$ is total on ground terms, 
(ii) $s \succ d$ for every domain element $d$ and every ground term $s$
that is not a domain element,
and 
(iii) $s \succ t$ for every ground FG term $s$ and
every ground BG term $t$.
These conditions are
easily satisfied by an LPO with an operator precedence
in which FG operators are larger than BG operators
and domain elements are minimal with, for example,
$\cdots \succ -2 \succ 2 \succ -1 \succ 1 \succ 0$ to achieve well-foundedness.

Condition (iii) and stability under \emph{simple} substitutions together
justify to always order $s \succ X$ where $s$
is a non-variable FG term and $X$ is an abstraction variable.
By contrast, $s \succ x$ can only hold if
$x \in \vars(s)$.
Intuitively, the combination of hierarchic reduction orderings
and abstraction variables affords ordering
more terms.

The ordering $\succ$ is extended to literals over
terms by identifying a positive literal $s \approx t$ with the multiset $\{s,t\}$,
a negative literal $s \not\approx t$ with $\{s,s,t,t\}$, and using the
multiset extension of $\succ$.  Clauses are compared by the multiset
extension of $\succ$, also denoted by~$\succ$.

The non-strict orderings $\succeq$ are defined as $s \succeq t$ if and only if $s \succ t$ or $s = t$
(the latter is multiset equality in case of literals and clauses).
A literal $L$ is called \emph{maximal} (\emph{strictly maximal})
\looseness=-1
in a clause $L \lor C$ if there is no $K \in C$ with $K \succ L$ ($K \succeq L$).

\section{Weak Abstraction}
\label{sec:weak-abstraction}
To refute an input set of $\Sigma$-clauses,
hierarchic superposition calculi derive BG clauses from them
and pass the latter to a BG prover.
In order to do this, some separation of the FG and BG
vocabulary in a clause is necessary.
The technique used for this separation is known as \emph{abstraction}:
One (repeatedly) replaces some term $q$ in a clause by a new variable and
adds a disequations to the
clause, so that $C[q]$ is converted into the equivalent clause
$\zeta \not\approx q \lor C[\zeta]$,
where $\zeta$ is a new (abstraction or ordinary) variable.

The calculus by Bachmair, Ganzinger, and
Waldmann~\cite{Bachmair:Ganzinger:Waldmann:TheoremProvingHierarchical:AAECC:94}
works on ``fully abstracted'' clauses:
Background terms occurring below a FG operator
or in an equation between a BG and a FG term
or vice versa
are abstracted out until one arrives at a clause in which no
literal contains both FG and BG operators.

A problematic aspect of any kind of abstraction is that it
tends to increase the number
of incomparable terms in a clause, which leads to an undesirable
growth of the search space of a theorem prover.
For instance, if we abstract out the subterms $t$ and $t'$ in a
ground clause $\op f(t) \approx \op g(t')$,
we get
${x \not\approx t} \lor {y \not\approx t'} \lor {\op f(x) \approx \op g(y)}$,
and the two new terms $\op f(x)$ and $\op g(y)$ are incomparable in
any reduction ordering.
In~\cite{Bachmair:Ganzinger:Waldmann:TheoremProvingHierarchical:AAECC:94}
this problem is mitigated
by considering only instances where BG-sorted variables
are mapped to BG terms:
In the terminology
of the current paper, all BG-sorted variables
in~\cite{Bachmair:Ganzinger:Waldmann:TheoremProvingHierarchical:AAECC:94}
have the kind ``abstraction''.
This means that, in the example above, we obtain the two terms
$\op f(X)$ and $\op g(Y)$.
If we use an LPO with a precedence in which $\op f$ is larger than $\op g$
and $\op g$ is larger than every BG operator,
then for every simple ground substitution $\tau$,
$\op f(X)\tau$ is strictly larger that
$\op g(Y)\tau$,
so we can still consider $\op f(X)$ as the only maximal term in the literal.

The advantage of full abstraction is that this clause structure
is preserved by all inference rules.
There is a serious drawback, however:
Consider the clause set
$
  N = \{\,1 + \op c \not\approx 1 + \op c\,\}
$.
Since $N$ is ground, we have $\sgi(N) = N$, and since $\sgi(N)$
is unsatisfiable, $N$ is trivially sufficiently complete.
Full abstraction turns $N$ into
$
  N' = \{\,X \not\approx \op c \,\lor\, 1 + X \not\approx 1 + X\,\}
$.
In the simple ground instances of $N'$, $X$ is mapped to all
pure BG terms.
However, there are $\Sigma$-interpretations of $\sgi(N')$
in which $\op c$ is interpreted differently from any pure BG term,
so $\sgi(N') \cup \GndTh(\BAlgs)$ does have a $\Sigma$-model
and $N'$ is no longer sufficiently complete.
In other words, 
the calculus
of~\cite{Bachmair:Ganzinger:Waldmann:TheoremProvingHierarchical:AAECC:94}
is refutationally complete for clause sets that
are fully abstracted and sufficiently complete,
but full abstraction may destroy sufficient
completeness.
(In fact, the calculus is not able to refute $N'$.)

The problem that we have seen is caused by the fact that
full abstraction replaces FG terms by abstraction variables,
which may not be mapped to FG terms later on.
The obvious fix would be to use ordinary variables instead of
abstraction variables whenever the term to be abstracted out
is not a pure BG term, but as we have seen above, this would
increase the number of incomparable terms and it would therefore
be detrimental to the performance of the prover.

Full abstraction is a property that is stronger than
actually necessary for the completeness proof
of~\cite{Bachmair:Ganzinger:Waldmann:TheoremProvingHierarchical:AAECC:94}.
In fact, it was claimed in a footnote
in~\cite{Bachmair:Ganzinger:Waldmann:TheoremProvingHierarchical:AAECC:94}
that the calculus could be optimized
by abstracting out
only non-variable BG terms
that occur below a FG operator.
This is incorrect, however: Using this abstraction rule,
neither our calculus nor the calculus
of~\cite{Bachmair:Ganzinger:Waldmann:TheoremProvingHierarchical:AAECC:94}
would be able to refute $\{\,1 + 1 \approx 2,$ $(1+1)+\op c \not\approx 2 + \op c\,\}$,
even though this set is unsatisfiable and trivially
sufficiently complete.
We need a slightly different abstraction rule to avoid this problem:

\begin{definition}
\label{def:weak-abs}
A BG term $q$ is a \emph{target term} in a clause $C$
if $q$ is neither a domain element nor a variable
and if $C$ has the form $C[f(s_1,\dots,q,\dots,s_n)]$,
where $f$ is a FG operator or at least one of the $s_i$
is a FG or impure BG term.\footnote{%
  Target terms are terms that need to be abstracted out; so
  for efficiency reasons, it is advantageous to keep the number of target
  terms as small as possible.
  We will show in Sect.~\ref{sec:refutational-completeness}
  why domain elements may be treated differently from other non-variable
  terms.
  On the other hand, all the results in the following sections
  continue to hold if the restriction that $q$ is not a domain element
  is dropped (\ie, if domain elements are abstracted out as well).
  We will make use of this fact in Sect.~\ref{sec:linear-arithmetic}.
}

A clause is called \emph{weakly abstracted} if it does not
have any target terms.

The \emph{weakly abstracted version of a clause} is the clause that
is obtained by exhaustively replacing $C[q]$ by
\begin{itemize}
\item
  $C[X] \lor X \not\approx q$, where $X$ is a new
  abstraction variable, if $q$ is a pure target term in $C$,
\item
  $C[y] \lor y \not\approx q$, where $y$ is a new
  ordinary variable, if $q$ is an impure\linebreak[3] target term in $C$.
\end{itemize}
The weakly abstracted version of a clause $C$ is denoted by $\abstr(C)$;
if $N$ is a set of clauses then $\abstr(N) = \{\,\abstr(C) \mid C \in N\,\}$. \qed
\end{definition}
For example, weak abstraction of the clause
$\op g(1, \alpha, \op f(1)+(\alpha+1), z) \approx \beta$ yields
$\op g(1, X, \op f(1)+Y, z) \approx \beta \vee X \not\approx \alpha \vee Y \not\approx \alpha+1$.
Note that the terms $1$, $\op f(1) + (\alpha+1)$,
$z$, and $\beta$ are not abstracted out:
$1$ is a domain element;
$\op f(1)+(\alpha+1)$ has a BG sort, but it is not
a BG term;
$z$ is a variable;
and $\beta$ is not a proper subterm of any other term.
The clause $\op{write}(\op{a}, 2, \op{read}(\op{a}, 1)+1) \approx \op b$
is already weakly abstracted.
Every pure BG clause is trivially weakly abstracted.

Nested abstraction is only necessary for certain impure BG terms.
For instance, the clause $\op f(z+\alpha) \approx 1$ has two target
terms, namely $\alpha$ (since $z$ is an impure BG term)
and $z+\alpha$ (since $\op f$ is a FG operator).
If we abstract out $\alpha$,
we obtain $\op f(z+X) \approx 1 \lor X \not\approx \alpha$.
The new term $z + X$ is still a target term,
so one more abstraction step yields
$\op f(y) \approx 1 \lor X \not\approx \alpha \lor y \not\approx z+X$.
(Alternatively, we can first abstract out $z+\alpha$,
yielding $\op f(y) \approx 1 \lor y \not\approx z+\alpha$,
and then $\alpha$. The final result is the same.)

It is easy to see that the abstraction process described in
Def.~\ref{def:weak-abs} does in fact terminate.
For any clause, we consider the
multiset of the numbers of non-variable occurrences in the left and
right-hand sides of its literals:
If we abstract out a target term $q$, then $q$ has $k \geq 1$
non-variable occurrences and it occurs as a subterm of a
left or right-hand side $s[q]$ with $n > k$ non-variable occurrences.
After the abstraction step, $s[\zeta]$ has
$n{-}k$ non-variable occurrences and the two terms $\zeta$ and $q$ in the
new abstraction literal have $0$ and $k$ non-variable occurrences.
Since $n{-}k$, $k$, and $0$ are strictly smaller than $n$,
the multiset decreases. Termination follows from the
well-foundedness of the multiset ordering.
\begin{proposition}
\label{prop:wab-equivalence-transformation}
If $N$ is a set of clauses
and $N'$ is obtained from $N$ by
replacing one or more clauses by their weakly abstracted versions,
then $\sgi(N)$ and $\sgi(N')$ are equivalent
and $N'$ is sufficiently complete
whenever $N$ is.
\end{proposition}
\begin{proof}
Let us first consider the case of a single abstraction step
applied to a single clause.
Let $C[q]$ be a clause with a target term $q$ and let
$D \,=\, C[\zeta] \lor \zeta \not\approx q$ be the result of abstracting
out $q$ (where $\zeta$ is a new abstraction variable, if $q$ is pure,
and a new ordinary variable, if $q$ is impure).
We will show that $\sgi(C)$ and $\sgi(D)$ have the same models.

In one direction let $I$ be an arbitrary model of $\sgi(C)$.
We have to show that $I$ is also a model of
every simple ground instance $D\tau$ of $D$.
If $I$ satisfies the disequation $\zeta\tau \not\approx q\tau$
then this is trivial.
Otherwise, $\zeta\tau$ and $q\tau$ have the same interpretation
in $I$.
Since $\dom(\tau) \supseteq \vars(D) = \vars(C) \cup \{\zeta\}$,
$C\tau$ is a simple ground instance of $C$,
so $I$ is a model of $C\tau = C\tau[q\tau]$.
By congruence, we conclude that $I$ is also a model of $C\tau[\zeta\tau]$,
hence it is a model of
\looseness=-1
$D\tau \,=\, C\tau[\zeta\tau] \lor \zeta\tau \not\approx q\tau$.

In the other direction let $I$ be an arbitrary model of $\sgi(D)$.
We have to show that $I$ is also a model of
every simple ground instance $C\tau$ of $C$.
Without loss of generality assume
that $\zeta \notin \dom(\tau)$.
If $\zeta$ is an abstraction variable, then
$q$ is a pure BG term, and since $\tau$ is a simple substitution,
$q\tau$ is a pure BG term as well.
Consequently, the substitutions
$[\zeta \mapsto q\tau]$ and $\tau' = \tau[\zeta \mapsto q\tau]$
are again simple substitutions and
$D\tau'$ is a simple ground instance of $D$.
This implies that $I$ is a model of $D\tau'$.
The clause $D\tau'$ has the form
$D\tau' \,=\, C\tau'[\zeta\tau'] \lor \zeta\tau' \not\approx q\tau'$;
since $\zeta\tau' = q\tau$, $C\tau' = C\tau$ and $q\tau' = q\tau$,
this is equal to
$C\tau[q\tau] \lor q\tau \not\approx q\tau$.
Obviously, the literal $q\tau \not\approx q\tau$ must be false in $I$,
so $I$ must be a model of $C\tau[q\tau] = C[q]\tau = C\tau$.

By induction over the number of abstraction steps we conclude
that for any clause $C$, $\sgi(C)$ and $\sgi(\abstr(C))$ are equivalent.
The extension to clause sets $N$ and $N'$ follows then from the fact that
$I$ is a model of $\sgi(N)$ if and only if it is a model of
$\sgi(C)$ for all $C \in N$.
Moreover, the equivalence of $\sgi(N)$ and $\sgi(N')$
implies obviously that $N'$ is
sufficiently complete whenever $N$ is.
\qed
\end{proof}

In contrast to full abstraction, the weak abstraction rule
does not require abstraction of FG terms
(which can destroy sufficient completeness
if done using abstraction variables, and which is detrimental to
the performance of a prover if done using ordinary variables).
BG terms are usually abstracted out using abstraction variables.
The exception are BG terms that are impure,
\ie, that contain ordinary variables themselves.
In this case, we cannot avoid to use ordinary variables
for abstraction,
otherwise, we might again destroy sufficient completeness.
For example, the clause set
$\{\, \op P(1+y),\ \neg \op P(1+\op c)\,\}$ is sufficiently complete.
If we used an abstraction variable instead of an
ordinary variable to abstract out the impure subterm $1+y$,
we would get
$\{\, \op P(X) \vee X \not\approx 1+y ,\ \neg \op P(1+\op c)\,\}$,
which is no longer
sufficiently complete.

In input clauses (that is, before abstraction),
BG-sorted variables may be declared
as ``ordinary'' or ``abstraction''.
As we have seen above,
using abstraction variables can reduce the search space;
on the other hand,
abstraction variables may be detrimental to sufficient completeness.
Consider the following example:
The set of clauses
$N = \{\,\neg \op f(x) > \op g(x) \lor \op h(x) \approx 1,$
$\neg \op f(x) \leq \op g(x) \lor \op h(x) \approx 2,$
$\neg \op h(x) > 0\,\}$
is unsatisfiable \wrt~linear integer arithmetic, but
since it is not sufficiently complete,
the hierarchic superposition calculus does not detect the unsatisfiability.
Adding the clause
$X > Y \lor X \leq Y$
to $N$ does not help: Since the abstraction variables $X$ and $Y$
may not be mapped to the FG terms
$\op f(x)$ and $\op g(x)$ in a simple ground instance,
the resulting set is still not
sufficiently complete.
However, if we add the clause
$x > y \lor x \leq y$,
the set of clauses becomes (vacuously)
sufficiently complete
and its unsatisfiability is detected.

One might wonder whether it is also possible to gain anything if
the abstraction process is performed using ordinary variables instead of
abstraction variables. The following proposition shows that this is
not the case:

\begin{proposition}
Let $N$ be a set of clauses, let $N'$ be the result of
weak abstraction of $N$ as defined above,
and let $N''$ be the result of weak abstraction of $N$
where all newly introduced variables are ordinary variables.
Then $\sgi(N')$ and $\sgi(N'')$ are equivalent
and $\sgi(N')$ is sufficiently complete
if and only if~$\sgi(N'')$~is.
\end{proposition}
\begin{proof}
By Prop.~\ref{prop:wab-equivalence-transformation},
we know already that $\sgi(N)$ and $\sgi(N')$ are equivalent.
Moreover, it is easy to check the
proof of Prop.~\ref{prop:wab-equivalence-transformation}
is still valid if we assume that the newly introduced variable $\zeta$
is always an ordinary variable.
(Note that the proof requires that abstraction variables are mapped
only to pure BG terms, but it does not require that a variable that
is mapped to a pure BG term must be an abstraction variable.)
So we can conclude in the same way
that $\sgi(N)$ and $\sgi(N'')$ are equivalent,
and hence, that $\sgi(N')$ and $\sgi(N'')$ are equivalent.
From this, we can conclude that
$N'$ is sufficiently complete whenever $N''$ is.
\qed
\end{proof}

In the rest of the paper we will need some technical lemmas
that relate clauses, their (partial) abstractions, instances of these
clauses, and the target terms in these clauses.

\begin{lemma}
\label{lemma:non-target-stable-under-simple-instances}%
Let $w$ be a subterm of a literal $L = [\neg]\, v[w] \approx v'$.
Let $\sigma$ be a simple substitution
and let $K$ be a literal $[\neg]\, v\sigma[w\sigma] \approx v''$.
If $w\sigma$ is a target term in $K$,
then $w$ is a variable or a target term in $L$.
\end{lemma}
\begin{proof}
If $w$ is neither a variable nor a target term in $L$,
then $w$ is a FG term or a domain element, or it equals $v$,
or it occurs below a BG operator $f$
and all other terms occurring below $f$ are pure BG terms.
Each of these properties is preserved by instantiation with a
simple substitution.
(Note that simple instances of pure BG terms are again pure BG terms.)
\qed
\end{proof}

\begin{lemma}
\label{lemma:target-stable-under-unabstraction}%
Let $w$ be a subterm of a literal $L = [\neg]\, v[w] \approx v'$.
Let $\sigma$ be a substitution
that maps all abstraction variables in its domain to pure BG terms
and all ordinary BG-sorted variables in its domain to impure BG terms,
and let $K$ be a literal $[\neg]\, v\sigma[w\sigma] \approx v''$.
If $w$ is a target term in $L$
then $w\sigma$ is a target term~in~$K$.
\end{lemma}
\begin{proof}
The term $w$ is a target term in $L$ if and only if
(i) $w$ is a BG term that is neither a domain element nor a variable,
and (ii) $w$ occurs in $L$ in a term $f(s_1,\dots,w,\dots,s_n)$,
where $f$ is a FG operator or at least one of the $s_i$
is a FG or impure BG term.
Obviously $w\sigma$ cannot be a domain element or a variable.
Moreover it is easy to check that
$\sigma$ maps
BG terms to BG terms,
impure BG terms to impure BG terms,
and FG terms to FG terms,
so both properties (i) and (ii) are preserved by $\sigma$.
\qed
\end{proof}

\begin{lemma}
\label{lemma:no-new-targets-by-abstraction}%
Let $D_0$ be a clause.
Let $\tau_0$ be a simple substitution such that
$D_0\tau_0$ is a ground instance of $D_0$ and
let $\rho_0$ be the identity substitution.
For $n \in \{0,\dots,k{-}1\}$ let
$D_{n+1}$ be clauses
obtained from $D_0$ by successively abstracting out target terms
as described in Def.~\ref{def:weak-abs}, that is,
let $q_n$ be a target term in $D_n = D_n[q_n]$ and
let $D_{n+1} = D_n[\zeta_n] \lor \zeta_n \not\approx q_n$
(where $\zeta_n$ is a new abstraction variable, if $q_n$ is
a pure target term, and a new ordinary variable otherwise).
Let $\tau_{n+1} = \tau_n[\zeta_n \mapsto q_n\tau_n]$
and $\rho_{n+1} = \rho_n[\zeta_n \mapsto q_n\rho_n]$.
Then the following properties hold:
\begin{enumerate}[\rm(1)]
\item
  If $n \in \{0,\dots,k\}$, then
  $\tau_n$ and $\rho_n$ are simple substitutions and
  $\rho_n$ maps all abstraction variables in its domain to pure BG terms
  and all ordinary variables in its domain to impure BG terms.
\item
  If $n \in \{0,\dots,k\}$, then
  $\tau_n = \rho_n\tau_0$.
\item
  If $n \in \{0,\dots,k\}$ then
  $D_n\tau_n$ is a ground instance of $D_n$
  and has the form
  $D_n\tau_n$ \,$=$\,
  $D_0\tau_0 \lor \bigvee_{0\leq i<n} q_i\tau_i \not\approx q_i\tau_i$,
  where the literals $q_i\tau_i \not\approx q_i\tau_i$
  are ground instances of the abstraction literals
  introduced so far.
\item
  If $n \in \{0,\dots,k\}$ and if $D'_n$ is the subclause of $D_n$
  that is obtained by dropping the
  abstraction literals introduced so far, then $D'_n\rho_n = D_0$.
\item
  If $n \in \{0,\dots,k{-}1\}$ and if the target
  term $q_n$ occurs in the subclause $D'_n$ of $D_n$,
  then there is a target term $\qbar_n$ in $D_0$ such that $q_n\tau_n = \qbar_n\tau_0$.
\end{enumerate}
\end{lemma}
\begin{proof}
Property (1) follows by induction from the fact that all the mappings
$[\zeta_n \mapsto q_n\tau_n]$ and $[\zeta_n \mapsto q_n\rho_n]$
are simple
and from the fact that
$[\zeta_n \mapsto q_n\rho_n]$ maps an abstraction variable
to a pure BG term
or an ordinary variable to an impure BG term.

Property (2) is obvious for $n = 0$.
By induction, we obtain
$\tau_{n+1}
= \tau_n[\zeta_n \mapsto q_n\tau_n]
= \rho_n\tau_0[\zeta_n \mapsto q_n\rho_n\tau_0]
= \rho_n[\zeta_n \mapsto q_n\rho_n]\tau_0
= \rho_{n+1}\tau_0$
as required.

Property (3) is again obvious for $n = 0$.
By induction, we obtain
$D_{n+1}\tau_{n+1}$
\,$=$\, $D_n\tau_{n+1}[\zeta_n\tau_{n+1}] \lor \zeta_n\tau_{n+1} \not\approx q_n\tau_{n+1}$
\,$=$\, $D_n\tau_n[q_n\tau_n] \lor q_n\tau_n \not\approx q_n\tau_n$
\,$=$\, $D_n\tau_n \lor q_n\tau_n \not\approx q_n\tau_n$
\,$=$\, $D_0\tau_0 \lor \bigvee_{0\leq i<n} q_i\tau_i \not\approx q_i\tau_i
         \lor q_n\tau_n \not\approx q_n\tau_n$
\,$=$\, $D_0\tau_0 \lor \bigvee_{0\leq i<n+1} q_i\tau_i \not\approx q_i\tau_i$
as required.
Since $q_n$ is a subterm of $D_n$ and
$D_n\tau_n$ is ground by induction,
we can conclude that $q_i\tau_n \not\approx q_i\tau_n$ is ground as well.

To prove property (4), we write $D_n$ in the form
$D_n \,=\, D'_n \lor E_n$, where $E_n$ is the subclause
consisting of all abstraction literals introduced so far.
For $n = 0$, there is nothing to prove.
If $q_n$ occurs in $D'_n$, then
$D_n = D'_n[q_n] \lor E_n$ and
$D_{n+1}
= D'_n[\zeta_n] \lor E_n \lor \zeta_n \not\approx q_n$,
therefore
$D'_{n+1} = D'_n[\zeta_n]$
and $D'_{n+1}\rho_{n+1} = D'_n\rho_{n+1}[\zeta_n\rho_{n+1}]
= D'_n\rho_n[q_n\rho_n]
= D'_n\rho_n = D_0$.
Otherwise $q_n$ occurs in $E_n$, then
$D_n = D'_n \lor E_n[q_n]$ and
$D_{n+1} = D'_n \lor E_n[\zeta_n] \lor \zeta_n \not\approx q_n$,
therefore
$D'_{n+1} = D'_n$
and $D'_{n+1}\rho_{n+1} = D'_n\rho_{n+1} = D'_n\rho_n = D_0$.

It remains to prove property (5).
By property (4) $\rho_n$ maps $D'_n$ to $D_0$;
since $q_n$ occurs in $D'_n$,
$q_n\rho_n$ is a subterm $\qbar_n$ of $D_0$.
By property (1) and Lemma~\ref{lemma:target-stable-under-unabstraction},
$\qbar_n$ must be a target term in $D_0$.
Property (2) yields
$\qbar_n\tau_0 = q_n\rho_n\tau_0 = q_n\tau_n$.
\qed
\end{proof}

\section{Base Inference System}
\label{sec:base-inference-system}
An \emph{inference system $\Inf$} is a set of inference rules. 
By an \emph{$\Inf$ inference} we mean an instance of an inference
rule from $\Inf$ such that all conditions are satisfied. 

The \emph{base inference system $\HSP_\Base$}
of the hierarchic superposition calculus
consists of the inference rules
\IR{Equality resolution}, \IR{Negative superposition}, \IR{Positive} \IR{superposition},
\IR{Equality factoring}, and \IR{Close} defined below.
The calculus is parameterized by a hierarchic reduction ordering $\succ$
and by a ``selection function''
that assigns to every
clause a (possibly empty) subset of its negative FG literals.
All inference rules are applicable only to
weakly abstracted premise clauses.

\begin{equation*}
\infrule[\IR{Equality resolution}]{
s \not\approx t \lor C
}{
\abstr(C\sigma)
}
\end{equation*}
if
(i) $\sigma$ is a simple mgu of $s$ and $t$,
(ii) $s\sigma$ is not a pure BG term, and
(iii) if the premise has selected literals, then
$s \not\approx t$ is selected in the premise, otherwise
$(s \not\approx t)\sigma$ is maximal in
$(s \not\approx t \lor C)\sigma$.\footnote{%
  As in~\cite{Bachmair:Ganzinger:Waldmann:TheoremProvingHierarchical:AAECC:94},
  it is possible to strengthen the maximality condition by requiring that there
  exists some simple ground substitution $\psi$ such that
  $(s \not\approx t)\sigma\psi$ is maximal in
  $(s \not\approx t \lor C)\sigma\psi$ (and analogously for the
  other inference rules).}

For example, \IR{Equality resolution} is applicable to $1+ \op c \not\approx 1+x$ with the
simple mgu $[x \mapsto \op c]$, but it is not applicable to $1+ \alpha \not\approx 1+x$, since $1+ \alpha$ is a
pure~BG~term. 

\begin{equation*}
\infrule[\IR{Negative superposition}]{
l \approx r \lor C \qquad s[u] \not\approx t \lor D
}{
\abstr((s[r] \not\approx t \lor C \lor D)\sigma)
} 
\end{equation*}
if
(i) $u$ is not a variable,
(ii) $\sigma$ is a simple mgu of $l$ and $u$, 
(iii) $l\sigma$ is not a pure BG term,
(iv) $r\sigma \not \succeq l\sigma$, 
(v) $(l \approx r)\sigma$ is strictly maximal in $(l \approx r \lor C)\sigma$,
(vi) the first premise does not have selected literals,
(vii) $t\sigma \not \succeq s\sigma$, and
(viii) if the second premise has selected literals, then
$s \not\approx t$ is selected in the second premise, otherwise
$(s \not\approx t)\sigma$ is maximal in $(s \not\approx t \lor D)\sigma$.

\begin{equation*}
\infrule[\IR{Positive superposition}]{
l \approx r \lor C \qquad s[u] \approx t \lor D
}{
\abstr((s[r] \approx t \lor C \lor D)\sigma)
} 
\end{equation*}
if
(i) $u$ is not a variable,
(ii) $\sigma$ is a simple mgu of $l$ and $u$, 
(iii) $l\sigma$ is not a pure BG term,
(iv) $r\sigma \not \succeq l\sigma$, 
(v) $(l \approx r)\sigma$ is strictly maximal in $(l \approx r \lor C)\sigma$,
(vi) $t\sigma \not \succeq s\sigma$,
(vii) $(s \not\approx t)\sigma$ is strictly maximal in $(s \approx t \lor D)\sigma$, and
(viii) none of the premises has selected literals.

\begin{equation*}
\infrule[\IR{Equality factoring}]{
s \approx t \lor l \approx r \lor C
}{
\abstr((l \approx r \lor t \not\approx r \lor C)\sigma)
} 
\end{equation*}
where 
(i) $\sigma$ is a simple mgu of $s$ and $l$, 
(ii) $s\sigma$ is not a pure BG term,
(iii) $(s \approx t)\sigma$ is maximal in $(s \approx t \lor l \approx r \lor C)\sigma$,
(iv) $t\sigma \not\succeq s\sigma$,
(v) $l\sigma \not\succeq r\sigma$, and
(vi) the premise does not have selected literals.

\begin{equation*}
\infrule[\IR{Close}]{
C_1 \quad \cdots \quad C_n
}{
\Box
} 
\end{equation*}
if $C_1,\ldots,C_n$ are BG clauses and $\{ C_1,\ldots,C_n\}$  is
$\BAlgs$-unsatisfiable,
\ie, no interpretation in $\BAlgs$ is a $\Sigma_\back$-model of $\{ C_1,\ldots,C_n\}$.

Notice that \IR{Close} is not restricted to take \emph{pure} BG clauses only. The
reason is that also impure BG clauses admit simple ground instances that are pure.

\begin{theorem}
\label{thm:hsp-base-sound}%
The inference rules of $\HSP_\Base$
are sound \wrt\ $\models_\BAlgs$,
\ie, for every inference with premises in $N$ and conclusion $C$,
we have $N \models_\BAlgs C$.
\end{theorem}
\begin{proof}
\IR{Equality resolution}, \IR{Negative superposition},
\IR{Positive superposition}, and \IR{Equality} \IR{factoring}
are clearly sound \wrt\ $\models$,
and therefore also sound \wrt\ $\models_\BAlgs$.
For \IR{Close}, soundness \wrt\ $\models_\BAlgs$ follows immediately
from the definition.
\qed
\end{proof}

All inference rules of $\HSP_\Base$ involve (simple) mgus. Because of the two kinds of
variables, abstraction and ordinary ones, the practical question arises if standard unification
algorithms can be used without or only little modification. For example, the
terms $Z$ and $(x+y)$ admit a simple mgu 
$\sigma = [x \mapsto X,\, y \mapsto Y,\, Z \mapsto X+Y]$. This prompts for the use of weakening substitutions as in
many-sorted logics with subsorts~\cite{Walther1988}.%
As we will see below, such
substitutions never need to be considered.

More precisely, we call a simple substitution $\sigma$ \emph{restricted} 
if $X\sigma \in \cX_\back \cup \Omega_\back^D$ for every abstraction variable $X \in \cX_\back$. That
is, a restricted simple substitution maps
every abstraction variable to an abstraction variable or a domain element.
\begin{proposition}
\label{prop:restricted-substitutions}
Let $s$ and $t$ be subterms of weakly abstracted clauses
such that $\sigma$ is a simple most general unifier of $s$ and $t$
and $s\sigma$ is not a pure BG term.
Then $\sigma$ is a restricted substitution.
\end{proposition}
\begin{proof}
Since $\sigma$ is a simple unifier of $s$ and $t$,
we have $s\sigma = t\sigma$.
Moreover, simple substitutions map pure BG terms to pure BG terms,
so we know that neither $s$ nor $t$ are pure BG terms.
This leaves the following possibilities:

Case 1: $s$ or $t$ is an ordinary BG variable $x$.
In this case, $\sigma$ has the form $[x \mapsto t]$ or $[x \mapsto s]$,
and since $Y\sigma = Y$ for every abstraction variable,
this is obviously a restricted substitution.

Case 2: $s$ or $t$ is a non-variable impure BG term.
Without loss of generality assume that $s$ is a non-variable impure BG term,
and that, if $t$ is also a non-variable impure BG term,
the depth of $s$ is not smaller than the depth of $t$
(otherwise swap $s$ and $t$).
The term $s$ has the form $f(s_1,\dots,s_k)$,
where $f$ is a BG operator symbol.
Since $s$ occurs in a weakly abstracted clause,
it may not contain any target terms,
so we have either that every $s_i$ is a variable or a domain element
and at least one $s_j$ is an ordinary variable,
or that exactly one $s_j$ is a non-variable impure BG term
and every $s_i$, $i \not= j$, is an abstraction variable or a domain element.
Let $p$ be the position of the deepest operator symbol in $s$
that is not a domain element.
We apply the decomposition rule
\[\frac{M \cup \{g(u_1,\dots,u_m) \doteq g(v_1,\dots,v_m)\}}
{M \cup \{u_1 \doteq v_1,\dots,u_m \doteq v_m\}}\]
exhaustively to the unifiability problem $\{s \doteq t\}$
and obtain
a unifiability problem $\{u_1 \doteq v_1, \dots, u_n \doteq v_n\}$.
Clearly $\sigma$ is also a most general unifier of $M$.

Case 2.1: $t\rstr{p} = g(\dots)$.
Then $u_1,\dots,u_k$ are direct subterms of $s\rstr{p}$, so they
are either BG variables or domain elements;
$u_{k+1},\dots,u_n$ are either BG abstraction variables or domain elements;
and $v_1,\dots,v_k$ are direct subterms of $t\rstr{p}$.

Case 2.1.1: $t$ is also an impure BG term.
Then $v_1,\dots,v_k$ are
also BG variables or domain elements,
and $v_{k+1},\dots,v_n$ are also BG abstraction variables or domain elements.
Consequently, $\sigma$ binds every variable to a variable or a domain element,
so it is restricted.

Case 2.1.2: $t$ is a FG term.
Then (every subterm of) every $v_i$ must be a variable, or a domain element,
or a FG term, otherwise it would be a target term.
So $\sigma$ cannot map an abstraction variable to
a BG term that is neither a variable nor a domain element,
hence it is again restricted.

Case 2.2: $t\rstr{qj} = \zeta$ for some proper prefix $q$ of $p$.
Then $u_1,\dots,u_k$ are direct subterms of $s\rstr{q}$, so
$u_j$ is a non-variable impure BG term and every $u_i$, $i \not= j$
is a BG abstraction variable or a domain element;
$u_{k+1},\dots,u_n$ are either BG abstraction variables or domain elements;
and $v_1,\dots,v_k$ are direct subterms of $t\rstr{q}$,
where $v_j = \zeta$.

Case 2.2.1: $v_j = \zeta$ is an abstraction variable $X$.
Then every equation $u_i \doteq v_i$ contains at least one pure BG term,
so $s\sigma = t\sigma$ is a pure BG term, contradicting the assumption.

Case 2.2.2: $v_j = \zeta$ is an ordinary variable $x$.
Since $v_i\sigma = u_i\sigma$ for every $i$,
each $v_i$ must be a BG term;
since $v_j$ is an impure BG term
this implies that every $v_i$, $i \not= j$, must
be a variable or domain element
(otherwise it would be a target term).
Assume that for some $l \not= j$, $v_l$ agrees with $x$.
Then $M = \{u_1 \doteq v_1, \dots, u_j \doteq x, \dots,
u_l \doteq x, \dots, u_n \doteq v_n\}$
is equivalent to
$M' = \{u_1 \doteq v_1, \dots, u_j \doteq u_l, \dots,
u_l \doteq x, \dots, u_n \doteq v_n\}$.
In $M'$, however, every equation contains at least one pure BG term,
so $s\sigma = t\sigma$ would be a pure BG term, contradicting the assumption.
Therefore we know that the variable $x$
occurs only once in $M$.
Consequently, $\sigma$ maps every variable except $x$ to a variable or
a domain element; since $x$ is an ordinary variable,
this implies that $\sigma$ is restricted.

Case 3: Both $s$ and $t$ are FG terms.
Since $s$ and $t$ occur in weakly abstracted clauses,
neither $s$ nor $t$ can contain a BG subterm that is not a variable
or a domain element.
Therefore $\sigma$ cannot map any pure variable $X$ to a BG term
that is not a variable or a domain element, so $\sigma$ is restricted.
\qed
\end{proof}

By condition (ii) in the \IR{Equality resolution} inference rule,
Prop.~\ref{prop:restricted-substitutions} guarantees that any simple mgu $\sigma$
of  an \IR{Equality resolution} inference is restricted. The same holds true
analogously for the other inference rules. As a consequence, unification algorithms
do not need to compute weakening substitutions.  
In essence, a suitably modified standard unification algorithm needs to prevent
binding an abstraction variable to a term other than an abstraction variable or a
domain element.  

For example, the (weakly abstracted) clause $s \not\approx t = Z+u \not\approx (x+y) + U$ admits 
no \IR{Equality resolution} inference. Although there is a simple 
mgu $\sigma = [x \mapsto X,\ y \mapsto Y,\ Z \mapsto X+Y,\ u \mapsto U]$ of $s$ and $t$, the term $s\sigma = (X+Y) + U$
is pure BG, hence violating condition (ii) in \IR{Equality resolution}, and so
$\sigma$ does not need to be computed in the first place.

In contrast to~\cite{Bachmair:Ganzinger:Waldmann:TheoremProvingHierarchical:AAECC:94},
the inference rules above include
an explicit weak abstraction in their
conclusion. Without it, conclusions would not be weakly abstracted in general.  For
example \IR{Positive superposition} applied to the weakly abstracted clauses $\op
f(X) \approx 1 \vee X \not\approx \alpha$ and $\op P(\op f(1) + 1)$ would then yield $\op P(1 + 1) \vee 1 \not\approx \alpha$,
whose $\op P$-literal is not weakly abstracted.
Additionally, the side conditions of our rules differ somewhat from
the corresponding rules
of~\cite{Bachmair:Ganzinger:Waldmann:TheoremProvingHierarchical:AAECC:94},
this is due on the one hand to the presence of impure BG terms
(which must sometimes be treated like FG terms),
and on the other hand to the fact that, after weak abstraction,
literals may still contain both FG and BG operators.

The inference rules are supplemented by a redundancy criterion,
that is, a mapping $\RCl$ from sets of formulae to sets of formulae
and a mapping $\RInf$ from sets of formulae to sets of inferences
that are meant to specify formulae that may be removed from $N$
and inferences that need not be computed.
($\RCl(N)$ need not be a subset of $N$
and $\RInf(N)$ will usually also contain inferences
whose premises are not in $N$.)

\COMMENT{UW}{12/01/2019}{
  Restored the traditional conditions ``$N \subseteq N'$ implies
  $\RInf(N) \subseteq \RInf(N')$''
  and
  ``$N' \subseteq \RCl(N)$ implies
  $\RCl(N) \subseteq \RCl(N \setminus N')$'' here
  and removed them in Thm.~\ref{thm:flat-superp-calc}.
  At a pinch, one can work without them, but the resulting conditions
  for fairness are not constructive.}

\begin{definition}
\label{dfn:redundancy-criterion}%
A pair $\Red = (\RInf,\RCl)$ is called a \emph{redundancy criterion}
(\emph{with respect to an inference system $\Inf$ and
a consequence relation $\models$}),
if the following conditions are satisfied for all sets of formulae $N$ and $N'$:
\begin{enumerate}[\rm(i)]
\item
$N \setminus \RCl(N) \models \RCl(N)$.
\item
If $N \subseteq N'$, then $\RCl(N) \subseteq \RCl(N')$
and $\RInf(N) \subseteq \RInf(N')$.
\item
If $\iota$ is an inference and its conclusion is in $N$, then $\iota \in \RInf(N)$.
\item
If $N' \subseteq \RCl(N)$, then $\RCl(N) \subseteq \RCl(N \setminus N')$
and $\RInf(N) \subseteq \RInf(N \setminus N')$.
\end{enumerate}
The inferences in $\RInf(N)$ and the formulae in $\RCl(N)$
are said to be \emph{redundant} with respect to~$N$. \qed
\end{definition}

Let $\SSP$ be
the ground standard superposition calculus
using the inference rules
equality resolution, negative superposition, positive superposition,
and equality factoring
(Bachmair and Ganzinger~\cite{Bachmair:Ganzinger:RewriteBasedTP:JLC:94},
Nieuwenhuis~\cite{Nieuwenhuis1991},
Nieuwenhuis and Rubio~\cite{Nieuwenhuis:Rubio:ParamodulationTheoremProving:HandbookAR:2001}).
To define a redundancy criterion for $\HSP_\Base$ and
to prove the refutational completeness of the calculus,
we use the same approach as
in~\cite{Bachmair:Ganzinger:Waldmann:TheoremProvingHierarchical:AAECC:94}
and relate $\HSP_\Base$ inferences to
the corresponding $\SSP$ inferences.

For a set of ground clauses $N$,
we define $\RClFlt(N)$ to be the set of all
clauses $C$ such that there exist clauses $C_1, \ldots, C_n \in N$ that
are smaller than $C$ with respect to $\succ$
and $C_1, \ldots, C_n \models C$.
We define $\RInfFlt(N)$ to be the set of all ground
$\SSP$ inferences $\iota$
such that either a premise of $\iota$ is in $\RClFlt(N)$ or else
$C_0$ is the conclusion of $\iota$ and
there exist clauses $C_1, \ldots, C_n \in N$ that
are smaller with respect to $\gtc$ than the maximal premise of $\iota$
and $C_1, \ldots, C_n \models C_0$.

The following results can be found in
\cite{Bachmair:Ganzinger:RewriteBasedTP:JLC:94}
and~\cite{Nieuwenhuis1991}:

\begin{theorem}
\label{thm:flat-superp-calc}%
The (ground) standard superposition calculus $\SSP$
and $\RedFlt = (\RInfFlt,\RClFlt)$
satisfy the following properties:
\leftmargini30pt
\begin{enumerate}[\rm(i)]
\item
$\RedFlt$
is a redundancy criterion with respect to $\models$.
\item
$\SSP$ together with $\RedFlt$
is refutationally complete.
\end{enumerate}
\end{theorem}

Let $\iota$ be an $\HSP_\Base$ inference
with premises $C_1,\ldots,C_n$ and conclusion $\abstr(C)$,
where the clauses $C_1,\ldots,C_n$ have no variables in common.
Let $\iota'$ be a ground $\SSP$ inference
with premises $C'_1,\ldots,C'_n$ and conclusion $C'$.
If $\sigma$ is a simple substitution such that $C' = C\sigma$
and $C'_i = C_i\sigma$ for all $i$,
and if none of the $C'_i$ is a BG clause,
then $\iota'$ is called a
\emph{simple ground instance} of $\iota$.
The set of all simple ground instances of an inference $\iota$ is denoted
by $\sgi(\iota)$.

\begin{definition}
\label{def:hierarchic-redundancy-criterion}
Let $N$ be a set of weakly abstracted clauses.
We define $\RInfSgi(N)$ to be the set of all inferences $\iota$ such that
either $\iota$ is not a \IR{Close} inference
and $\sgi(\iota) \subseteq \RInfFlt(\sgi(N)\cup\GndTh(\BAlgs))$,
or else
$\iota$ is a \IR{Close} inference and $\Box \in N$.
We define $\RClSgi(N)$ to be the set of all weakly abstracted clauses $C$
such that $\sgi(C) \subseteq \RClFlt(\sgi(N)\cup\GndTh(\BAlgs))
\cup \GndTh(\BAlgs)$.\footnote{%
 In contrast to~\cite{Bachmair:Ganzinger:Waldmann:TheoremProvingHierarchical:AAECC:94},
 we include $\GndTh(\BAlgs)$ in the redundancy criterion.
 (This is independent of the abstraction method used;
 it would also have been useful
 in~\cite{Bachmair:Ganzinger:Waldmann:TheoremProvingHierarchical:AAECC:94}.)
}
\qed
\end{definition}

\section{Refutational Completeness}
\label{sec:refutational-completeness}
To prove that $\HSP_\Base$ and $\RedSgi = (\RInfSgi, \RClSgi)$
are refutationally complete for sets of weakly abstracted $\Sigma$-clauses
and compact BG specifications $(\Sigma_\back,\BAlgs)$,
we use the same technique as
in~\cite{Bachmair:Ganzinger:Waldmann:TheoremProvingHierarchical:AAECC:94}:

First we show that $\RedSgi$ is a redundancy criterion
with respect to $\models_\BAlgs$,
and that a set of clauses remains
sufficiently complete if
new clauses are added or if
redundant clauses are deleted.%
The proofs for both properties are similar to the corresponding ones
in~\cite{Bachmair:Ganzinger:Waldmann:TheoremProvingHierarchical:AAECC:94};
the differences are due, on the one hand, to the fact that we include
$\GndTh(\BAlgs)$ in the redundancy criterion
and in the definition of sufficient completeness,
and, on the other hand, to the explicit abstraction steps in our
inference rules.

\begin{lemma}
\label{lemma:modelssgi-implies-modelsbalgs}
If $\sgi(N) \cup\GndTh(\BAlgs) \models \sgi(C)$, then $N \models_\BAlgs C$.
\end{lemma}
\begin{proof}
Suppose that $\sgi(N) \cup\GndTh(\BAlgs) \models \sgi(C)$.
Let $I'$ be a $\Sigma$-model of $N$ whose restriction to
$\Sigma_\back$ is contained in $\BAlgs$.
Clearly, $I'$ is also a model of $\GndTh(\BAlgs)$.
Since $I'$ does not add new elements to the sorts of $I = I'\rstr{\Sigma_\back}$
and $I$ is a term-generated $\Sigma_\back$-interpretation,
we know that for every ground $\Sigma$-term~$t'$ of a BG sort
there is a ground BG term $t$
such that $t$ and $t'$ have the same interpretation in $I'$.
Therefore, for every ground substitution $\sigma'$ there is an
equivalent simple ground substitution $\sigma$;
since $C\sigma$ is valid in $I'$, $C\sigma'$ is also valid.
\qed
\end{proof}

We call the simple most general unifier $\sigma$ that is computed during
an inference $\iota$
and applied to the conclusion the pivotal substitution
of $\iota$.
(For ground inferences,
the pivotal substitution is the identity mapping.)
If $L$ is the literal $[\neg]\, s \approx t$ or $[\neg]\, s[u] \approx t$
of the second or only premise
that is eliminated in $\iota$,
we call $L\sigma$ the pivotal literal of~$\iota$,
and we call $s\sigma$ or $s[u]\sigma$ the pivotal term of~$\iota$.


\begin{lemma}
\label{lemma:abstr-lits-after-inf-small-enough}%
Let $\iota$ be an $\HSP_\Base$ inference
\[
\infrule{C_1}{\abstr(C_0\sigma)}
\qquad
\textrm{or}
\qquad
\infrule{C_2 \qquad C_1}{\abstr(C_0\sigma)}
\]
from weakly abstracted premises
with pivotal substitution $\sigma$.
Let $\iota'$ be a simple ground instance of $\iota$ of the form
\[
\infrule{C_1\tau}{C_0\sigma\tau}
\qquad
\textrm{or}
\qquad
\infrule{C_2\tau \qquad C_1\tau}{C_0\sigma\tau}
\]
Then there is a simple ground instance of $\abstr(C_0\sigma)$
that has the form $C_0\sigma\tau \lor E$, where
$E$ is a (possibly empty) disjunction of literals $s \not\approx s$,
and each literal of $E$ is smaller than the
pivotal literal of $\iota'$.
\end{lemma}
\begin{proof}
Without loss of generality, we may assume that $\sigma$ is
an idempotent most general unifier over $\vars(C_1) \cup \vars(C_2)$
and that $\tau = \sigma\rho$
for some substitution $\rho$;
hence $\sigma\tau = \sigma\sigma\rho = \sigma\rho = \tau$
and in particular $C_0\sigma\tau = C_0\tau$.

Let $\tau_0 = \tau$ and
let $D_0 = C_0\sigma$.
Let $D_n$ be the result of the $n$th abstraction step
starting from $D_0$ for $n \in \{0,\dots,k\}$,
and let $\abstr(C_0\sigma) = D_k$.
According to part (3) of Lemma~\ref{lemma:no-new-targets-by-abstraction},
there are simple substitutions $\tau_n$ for $1 \leq n \leq k$
such that $D_n\tau_n$ is a ground instance of $D_n$
and has the form
$D_n\tau_n$ \,$=$\,
$D_0\tau_0 \lor \bigvee_{0\leq i<n} q_i\tau_i \not\approx q_i\tau_i$,
where the literals $q_i\tau_i \not\approx q_i\tau_i$
are ground instances of the abstraction literals
introduced so far
and $q_i$ is a target term in $D_i$.
Since $C_0\sigma\tau = D_0\tau = D_0\tau_0$,
this clause has essentially the required form.
We still have to prove, though, that
each literal $q_i\tau_i \not\approx q_i\tau_i$
is smaller than the
pivotal literal $L$ of $\iota'$.

If $q_i$ occurs in $D_i$ in a literal $K$ that has been
generated by the $j$th abstraction step ($j < i$),
then $K\tau_j$ is a literal $q_j\tau_j \not\approx q_j\tau_j$
and $q_i\tau_i$ is a proper subterm of $q_j\tau_j$.
By induction on the number of abstraction steps we obtain
$q_i\tau_i \not\approx q_i\tau_i \,\prec\, q_j\tau_j \not\approx q_j\tau_j
\,\prec\, L$.

If $q_i$ occurs in $D_i$ in a literal that has not been
generated by a previous abstraction step,
then by part (5) of Lemma~\ref{lemma:no-new-targets-by-abstraction}
there is a target term $\qbar_i$ in $D_0 = C_0\sigma$
such that $q_i\tau_i = \qbar_i\tau_0 = \qbar_i\tau$.
The term $\qbar_i$ is either a term from $C_1\sigma$, or a term from $C_2\sigma$,
or a subterm of the term $s[r]\sigma$
(the last two cases are possible only for
superposition inferences).
We analyze these cases separately.

Case 1:
$\qbar_i$ is a proper subterm of a term $v\sigma$,
where $[\neg] v\sigma \approx v''$ is a literal of $C_0\sigma$
and $[\neg] v \approx v'$ is a literal of $C_1$.

Case 1.1: $C_1$ does not have selected literals.
In this case, every literal of $C_1\tau$ is smaller than or equal to
the pivotal literal $L$ of $\iota'$.
Consequently, we obtain
$q_i\tau_i = \qbar_i\tau \prec v\sigma\tau = v\tau$,
thus
$q_i\tau_i \not\approx q_i\tau_i \,\prec\,
[\neg]\, v\tau \approx v'\tau \,\preceq\, L$.

Case 1.2: $C_1$ has selected literals.
In this case, $L$ is a selected literal in $C_1\tau$
and $\iota$ and $\iota'$ are either \IR{Negative superposition}
or \IR{Equality resolution} inferences.
Since $\qbar_i$ is a target term in $C_0\sigma$,
it is not a variable.

If $\qbar_i$ occurs in $v\sigma$ below a variable position of $v$,
then there is a $\zeta \in \dom(\sigma) \cap \vars(C_1)$
such that $\qbar_i$ is a subterm of $\zeta\sigma$.
Otherwise, there is a subterm $w$ of $v$ such that $\qbar_i = w\sigma$.
Since $C_1$ is weakly abstracted, $w$ cannot be a target term,
so by Lemma~\ref{lemma:non-target-stable-under-simple-instances},
$w$ must be a variable $\zeta$,
and again $\zeta \in \dom(\sigma) \cap \vars(C_1)$.
So in both cases, $\qbar_i$ is a subterm of $\zeta\sigma$,
and consequently, $q_i\tau_i = \qbar_i\tau$ is a subterm of $\zeta\sigma\tau = \zeta\tau$.

Let us first consider a \IR{Negative superposition} inference
operating on the literal $s[u] \not\approx t$ in $C_1$,
where $L \,=\, (s[u] \not\approx t)\tau$.
As $\sigma$ is a simple most general unifier of $l$ and $u$
and $\zeta \in \dom(\sigma) \cap \vars(C_1)$, we know that
$\zeta$ must occur in $u$. Moreover $u$ is not a variable,
so $\zeta$ is a proper subterm of $u$.
This implies $q_i\tau_i \preceq \zeta\tau \prec u\tau \preceq s[u]\tau$,
hence $q_i\tau_i \not\approx q_i\tau_i \,\prec\,
(s[u] \not\approx t)\tau \,=\, L$.

Otherwise $\iota$ is an \IR{Equality resolution} inference
operating on the literal $s \not\approx t$ in $C_1$,
where $L \,=\, (s \not\approx t)\tau$.
As $\sigma$ is a simple most general unifier of $s$ and $t$
and $\zeta \in \dom(\sigma) \cap \vars(C_1)$, we know that
$\zeta$ must occur in $s$ or $t$.
Assume without loss of generality that $\zeta$ occurs in $s$.
If $s = \zeta$, then
by the restrictions on selection functions
$t$ must be a FG term,
so $\zeta\tau = t\tau$ is a FG term as well.
Since the BG term $q_i\tau_i$ is a subterm of the FG term $\zeta\tau$,
it must be a proper subterm,
hence $q_i\tau_i \prec \zeta\tau = t\tau$
and $q_i\tau_i \not\approx q_i\tau_i \,\prec\,
(s \not\approx t)\tau \,=\, L$.
Otherwise $\zeta$ is a proper subterm of $s$, then
$q_i\tau_i \preceq \zeta\tau \prec s\tau = t\tau$
and we obtain again $q_i\tau_i \not\approx q_i\tau_i \,\prec\,
(s \not\approx t)\tau \,=\, L$.

Case 2:
$\qbar_i$ is a proper subterm of a term $v\sigma$,
where $[\neg] v\sigma \approx v''$ is a literal of $C_0\sigma$
and $[\neg] v \approx v'$ is a literal of $C_2$.
This case is proved similarly to case~1.1 above:
By the structure of $\SSP$ inferences, the literal $l\tau \approx r\tau$ of
$C_2\tau$ that has been used to replace $s\tau[l\tau]$ by $s\tau[r\tau]$
in the pivotal literal
is strictly maximal in $C_2\tau$ and $l\tau \succ r\tau$.
Consequently, we obtain
$q_i\tau_i = \qbar_i\tau \prec v\sigma\tau = v\tau \preceq l\tau
\preceq s\tau[l\tau]$
and thus
$q_i\tau_i \not\approx q_i\tau_i \,\prec\,
[\neg]\, s\tau[l\tau] \approx t\tau \,=\, L$.

Case 3:
It remains to consider the case that $\qbar_i$ is a subterm
of the term $s[r]\sigma$ produced in a superposition inference.
Then $q_i\tau_i = \qbar_i\tau \prec s[r]\sigma\tau = s[r]\tau
\prec s[l]\tau$,
hence
$q_i\tau_i \not\approx q_i\tau_i \,\prec\,
[\neg]\, s[l]\tau \approx t\tau \,=\, L$.
\qed
\end{proof}

As $M \subseteq M'$ implies $\RInfFlt(M) \subseteq \RInfFlt(M')$, we obtain
$\RInfFlt(\sgi(N) \setminus \sgi(N')) \subseteq
\RInfFlt(\sgi(N \setminus N'))$.
Furthermore, it is fairly easy to see that
$\sgi(N) \setminus (\RClFlt(\sgi(N)\cup\GndTh(\BAlgs)) \cup\GndTh(\BAlgs))
\subseteq \sgi(N \setminus \RClSgi(N))$.
Using these two results we can prove the following lemmas:

\begin{lemma}
\label{lemma:redsgi-redundancy-criterion}%
$\RedSgi = (\RInfSgi, \RClSgi)$ is a redundancy criterion with respect to
$\models_\BAlgs$.
\end{lemma}
\begin{proof}
We have to check the four conditions of
Def.~\ref{dfn:redundancy-criterion}.
The proof of property (ii) is rather trivial.
To check property~(i)
let $D$ be an arbitrary clause from $\sgi(\RClSgi(N))$.
Consequently, $D \in \RClFlt(\sgi(N)\cup\GndTh(\BAlgs)) \cup \GndTh(\BAlgs)$.
If $D \in \GndTh(\BAlgs)$, then 
trivially $\sgi(N\setminus\RClSgi(N))\cup\GndTh(\BAlgs) \models D$.
Otherwise
$D \in \RClFlt(\sgi(N)\cup\GndTh(\BAlgs))$,
and this implies
$\sgi(N)\cup\GndTh(\BAlgs)\setminus\RClFlt(\sgi(N)\cup\GndTh(\BAlgs))
\models D$.
Since
$\sgi(N)\cup\GndTh(\BAlgs)\setminus\RClFlt(\sgi(N)\cup\GndTh(\BAlgs))
\subseteq
\sgi(N)\setminus\RClFlt(\sgi(N)\cup\GndTh(\BAlgs))\cup\GndTh(\BAlgs)
=
\sgi(N)\setminus(\RClFlt(\sgi(N)\cup\GndTh(\BAlgs))\cup\GndTh(\BAlgs))\cup\GndTh(\BAlgs)
\subseteq
\sgi(N\setminus\RClSgi(N))\cup\GndTh(\BAlgs)$,
we obtain again
$\sgi(N\setminus\RClSgi(N))\cup\GndTh(\BAlgs) \models D$.
By Lemma~\ref{lemma:modelssgi-implies-modelsbalgs},
we can conclude that $N \setminus \RClSgi(N) \models_\BAlgs \RClSgi(N)$.

Condition~(iii) is obviously satisfied for all \IR{Close} inferences.
Suppose that $\iota$ is not a \IR{Close} inference and
its conclusion $\concl(\iota) = \abstr(C_0)$ is in $N$.
Showing that $\iota \in \RInfSgi(N)$ amounts to proving that
every simple ground instance of $\iota$
is redundant \wrt~$\sgi(N) \cup \GndTh(\BAlgs)$.
Let $\iota'$ be such a simple ground instance
with maximal premise $C_1\tau$
and conclusion
$C_0\tau$.
By Lemma~\ref{lemma:abstr-lits-after-inf-small-enough},
there is a simple ground instance of $\abstr(C_0)$
that has the form $C_0\tau \lor E$, where
$E$ is a (possibly empty) disjunction of literals $s \not\approx s$,
and each literal of $E$ is smaller than the
pivotal literal of $\iota'$.

By the structure of superposition inferences,
the clause $C_0\tau$ is obtained from $C_1\tau$ by replacing
the pivotal literal in $C_1\tau$
by (zero or more) smaller literals.
Since the literals in $E$ are also smaller than the
pivotal literal,
$C_0\tau \lor E$ is still smaller than $C_1\tau$.
Moreover, $C_0\tau \lor E$ entails $C_0\tau$,
so $\iota' \in \RInfFlt(\sgi(N) \cup \GndTh(\BAlgs))$.
As $\sgi(\iota) \subseteq \RInfFlt(\sgi(N) \cup \GndTh(\BAlgs))$,
the inference $\iota$ is contained in $\RInfSgi(N)$.
This proves condition~(iii).

We come now to the proof of condition~(iv).
Note that $N' \subseteq \RClSgi(N)$ implies
$\sgi(N') \subseteq \RClFlt(\sgi(N)\cup\GndTh(\BAlgs)) \cup\GndTh(\BAlgs)$,
and thus
$\sgi(N') \setminus \GndTh(\BAlgs)
\subseteq \RClFlt(\sgi(N)\cup\GndTh(\BAlgs))$.
If $\iota \in \RInfSgi(N)$
is a \IR{Close} inference,
then $\Box \in N$; since $\Box \notin \RClSgi(N)$,
$\iota$ is contained in $\RInfSgi(N \setminus N')$.
Otherwise $\sgi(\iota) \subseteq \RInfFlt(\sgi(N)\cup\GndTh(\BAlgs)) \subseteq
\RInfFlt(\sgi(N)\cup\GndTh(\BAlgs) \setminus (\sgi(N')\setminus\GndTh(\BAlgs)))
=
\RInfFlt(\sgi(N)\setminus\sgi(N')\cup\GndTh(\BAlgs))
\subseteq
\RInfFlt(\sgi(N\setminus N')\cup\GndTh(\BAlgs))$,
hence $\iota$ is again contained in
$\RInfSgi(N \setminus N')$.
Therefore
$\RInfSgi(N) \subseteq \RInfSgi(N \setminus N')$.
Analogously,
we can show that
$\RClSgi(N) \subseteq \RClSgi(N \setminus N')$.
\qed
\end{proof}

\begin{lemma}
Let $N$, $N'$ and $M$ be sets of weakly abstracted clauses
such that $N' \subseteq \RClSgi(N)$.
If $N$ is sufficiently complete, then
so are $N \cup M$ and $N \setminus N'$.
\end{lemma}
\begin{proof}
The sufficient completeness of $N \cup M$ is obvious;
the sufficient completeness of $N \setminus N'$ is proved in a similar
way as in part (i) of
the proof of Lemma~\ref{lemma:redsgi-redundancy-criterion}.
\qed
\end{proof}

We now encode arbitrary term-generated $\Sigma_\back$-interpretation
by sets of unit ground clauses in the following way:
Let $I \in \BAlgs$ be a term-generated $\Sigma_\back$-inter\-pre\-ta\-tion.
For every $\Sigma_\back$-ground term $t$ let $m(t)$ be the smallest ground term
of the congruence class of $t$ in $I$.
We define a rewrite system $\EA'$ by
$\EA' = \{\,t \to m(t) \mid t \in \TSigma,\, t \not= m(t)\,\}$.
Obviously, $\EA'$ is right-reduced;
since all rewrite rules are contained in $\succ$,
$\EA'$ is terminating;
and since every ground term $t$ has $m(t)$ as its only normal form,
$\EA'$ is also confluent.
Now let $\EA$ be the set of all rules $l \to r$ in $\EA'$
such that $l$ is not reducible by $\EA' \setminus \{l \to r\}$.
Clearly every term that is reducible by $\EA$ is also reducible by $\EA'$;
conversely every term that is reducible by $\EA'$ has a minimal subterm
that is reducible by $\EA'$ and the rule in $\EA'$ that is used to rewrite
this minimal subterm is necessarily contained in $\EA$.
Therefore $\EA'$ and $\EA$ define the same set
of normal forms, and from this we can conclude
that $\EA$ and $\EA'$ induce the same equality relation
on ground $\Sigma_\back$-terms.
We identify $\EA$ with the set of clauses
$\{\, t \approx t' \mid t \to t' \in \EA \,\}$.
Let $\DA$ be the set of
all clauses $t \not\approx t'$,
such that $t$ and $t'$ are distinct ground $\Sigma_\back$-terms
in normal form with respect to~$\EA$.\footnote{%
  Typically, $\EA$ contains two kinds of clauses, namely clauses
  that evaluate non-constant BG terms, such as $2 + 3 \approx 5$,
  and clauses that map parameters to domain elements, such as
  $\alpha \approx 4$.}


\begin{lemma}
\label{lemma:eada}%
Let $I \in \BAlgs$ be a term-generated $\Sigma_\back$-interpretation
and let $C$ be a ground BG clause.
Then $C$ is true in $I$ if and only if
there exist clauses $C_1,\ldots,C_n$ in $\EADA$
such that $C_1,\ldots,C_n \models C$
and $C \succeq C_i$ for $1 \leq i \leq n$.
\end{lemma}
\begin{proof}
The ``if'' part follows immediately from the fact that $I \models \EADA$.
For the ``only if'' part assume that the ground BG clause $C$ is
true in $I$. Consequently, there is some literal $s \approx t$ or
$s \not\approx t$ of $C$ that is true in $I$.
Then this literal follows from
(i) the rewrite rules in $\EA$ that are used to normalize $s$ to its
normal form $s'$,
(ii) the rewrite rules in $\EA$ that are used to normalize $t$ to its
normal form $t'$,
and, in the case of a negated literal $s \not\approx t$,
(iii) the clause $s' \not\approx t' \in \DA$.
It is routine to show that all these clauses are smaller than
or equal to $s \approx t$ or $s \not\approx t$, respectively,
and hence smaller than or equal to $C$.
\qed
\end{proof}

\begin{corollary}
\label{cor:eada-gndth}%
Let $I \in \BAlgs$ be a term-generated $\Sigma_\back$-interpretation.
Then $\EADA \models \GndTh(\BAlgs)$.
\end{corollary}
\begin{proof}
Since $I \in \BAlgs$, we have $I \models \GndTh(\BAlgs)$,
hence $\EADA \models \GndTh(\BAlgs)$ by Lemma~\ref{lemma:eada}.
\qed
\end{proof}

Let $N$ be a set of weakly abstracted clauses
and $I \in \BAlgs$ be a term-generated $\Sigma_\back$-interpretation, then
$N_I$ denotes the set
$\EADA \cup \{\, C\sigma \mid
\sigma$~simple, reduced with respect to $\EA$, $C \in N$,
$C\sigma$~ground$\,\}$.

\begin{lemma}
\label{lemma:RInfFltSgiNI}%
If $N$ is a set of weakly abstracted clauses, then
$\RInfFlt(\sgi(N)\cup\GndTh(\BAlgs)) \subseteq \RInfFlt(N_I)$.
\end{lemma}
\begin{proof}
By part (i) of Thm.~\ref{thm:flat-superp-calc} we have obviously
$\RInfFlt(\sgi(N)) \subseteq \RInfFlt(\EADA \cup \sgi(N)\cup\GndTh(\BAlgs))$.
Let $C$ be a clause in $\EADA \cup \sgi(N) \cup\GndTh(\BAlgs)$ and not in $N_I$.
If $C \in \GndTh(\BAlgs)$, then it is true in $I$,
so by Lemma~\ref{lemma:eada} it is either contained in $\EADA \subseteq N_I$
or it follows from smaller clauses in $\EADA$
and is therefore in $\RClFlt(\EADA \cup \sgi(N))$.
If $C \notin \GndTh(\BAlgs)$,
then $C = C'\sigma$ for some $C' \in N$, so
it follows from $C'\rho$ and $\EADA$,
where $\rho$ is the substitution that maps every variable $\zeta$
to the $\EA$-normal form of $\zeta\sigma$.
Since $C$ follows from smaller clauses in $\EADA \cup \sgi(N)$,
it is in $\RClFlt(\EADA \cup \sgi(N))$.
Hence $\RInfFlt(\EADA \cup \sgi(N) \cup\GndTh(\BAlgs)) \subseteq \RInfFlt(N_I)$.
\qed
\end{proof}

A clause set $N$ is called \emph{saturated (with respect to an inference system
$\Inf$ and a redundancy criterion $\Red$)} if $\iota \in \RInf(N)$ for every inference
\looseness=-1
$\iota$ with premises in~$N$.

\begin{theorem}
\label{thm:na-saturation}%
Let $I \in \BAlgs$ be a term-generated $\Sigma_\back$-interpretation
and let $N$ be a set of weakly abstracted $\Sigma$-clauses.
If $I$ satisfies all BG clauses in $\sgi(N)$
and $N$ is saturated with respect to
$\HSP_\Base$ and $\RedSgi$, then
$N_I$ is saturated with respect to
$\SSP$ and $\RedFlt$.
\end{theorem}
\begin{proof}
We have to show that every $\SSP$-inference from clauses of $N_I$
is redundant with respect to $N_I$, \ie, that it 
is contained in $\RInfFlt(N_I)$.
We demonstrate this in detail for the equality resolution and the
negative superposition rule.
The analysis of the other rules is similar.
Note that by Lemma~\ref{lemma:eada}
every BG clause that is true in $I$ and is not contained
in $\EADA$ follows from smaller clauses in $\EADA$,
thus it is in $\RClFlt(N_I)$;
every inference involving such a clause is in $\RInfFlt(N_I)$.

The equality resolution rule is obviously not applicable to clauses
from $\EADA$.
Suppose that $\iota$ is an equality resolution inference with a premise
$C\sigma$, where $C \in N$ and $\sigma$ is
a simple substitution and reduced with respect to $\EA$.
If $C\sigma$ is a BG clause, then $\iota$ is in $\RInfFlt(N_I)$.
If the pivotal term of $\iota$ is a pure BG term then the pivotal literal is pure BG as
well. Because the pivotal literal is maximal in $C\sigma$ it follows from properties of
the ordering that $C\sigma$ is a BG clause. Because we have already considered this case
we can assume from now on that the pivotal term of $\iota$ is not pure BG and that $C\sigma$
is an FG clause. It follows that $\iota$ is a simple ground instance of a
hierarchic inference $\iota'$ from $C$.
Since $\iota'$ is in $\RInfSgi(N)$, $\iota$ is in $\RInfFlt(\sgi(N)\cup\GndTh(\BAlgs))$,
by Lemma~\ref{lemma:RInfFltSgiNI}, this implies again $\iota \in \RInfFlt(N_I)$.

Obviously a clause from $\DA$ cannot be the first premise of a
negative superposition inference.
Suppose that the first premise is a clause from $\EA$.
The second premise
cannot be a FG clause, since the
maximal sides of maximal literals in a FG clause
are reduced;
as it is a BG clause, the inference is redundant.
Now suppose that $\iota$ is a negative superposition inference
with a first premise $C\sigma$, where $C \in N$ and $\sigma$ is
a simple substitution and reduced with respect to $\EA$.
If $C\sigma$ is a BG clause, then $\iota$ is in $\RInfFlt(N_I)$.
Otherwise, with the same arguments as for the equality resolution case above, the
pivotal term is not pure BG and $C\sigma$ is a FG clause.
Hence we can conclude that the second premise can be written as
$C'\sigma$, where $C' \in N$ is a FG clause (without loss of generality, $C$ and $C'$
do not have common variables).
If the overlap takes place below a variable occurrence,
the conclusion of the inference follows from $C\sigma$ and some
instance $C'\rho$, which are both smaller than $C'\sigma$.
Otherwise, $\iota$ is a simple ground instance of a
hierarchic inference $\iota'$ from $C$.
In both cases, $\iota$ is contained in $\RInfFlt(N_I)$.
\qed
\end{proof}

The crucial property of abstracted clauses
that is needed in the proof of this theorem
is that there are no superposition inferences between
clauses in $\EA$ and FG ground instances $C\sigma$ in $N_I$,
or in other words,
that all FG terms occurring in ground instances
$C\sigma$ are reduced \wrt~$\EA$.
This motivates the definition of target terms in Def.~\ref{def:weak-abs}:
Recall that two different domain elements must always be interpreted
differently in $I$ and that a domain element is smaller
in the term ordering than
any ground term that is not a domain element.
Consequently,
any domain element is the smallest term in its congruence class,
so it is reduced by~$\EA$.
Furthermore, by the definition of $N_I$,
$\zeta\sigma$ is reduced by~$\EA$ for every variable $\zeta$.
So variables and domain elements never need to be abstracted out.
Other BG terms (such as parameters $\alpha$ or
non-constant terms $\zeta_1 + \zeta_2$) have to be abstracted out if they occur
below a FG operator, or if one of their sibling terms is a FG term
or an impure BG term (since $\sigma$ can map the latter to a FG term).
On the other hand,
abstracting out FG terms
as in~\cite{Bachmair:Ganzinger:Waldmann:TheoremProvingHierarchical:AAECC:94}
is never necessary to ensure
that FG terms are reduced \wrt~$\EA$.

If $N$ is saturated with respect to
$\HSP_\Base$ and $\RedSgi$ and does not contain the empty clause,
then \IR{Close} cannot be applicable to $N$.
If
$(\Sigma_\back,\BAlgs)$ is compact,
this implies that
there is some term-generated $\Sigma_\back$-interpretation $I \in \BAlgs$
that satisfies all BG clauses in $\sgi(N)$.
Hence, by Thm.~\ref{thm:na-saturation},
the set of \emph{reduced simple} ground instances of $N$ has a model
that also satisfies $\EADA$.
Sufficient completeness
allows us to show that this is in fact a model of
\emph{all} ground instances of clauses in~$N$
and that $I$ is its restriction to~$\Sigma_\back$:

\begin{theorem}
\label{thm:hiersup-ref-complete}%
If the BG specification $(\Sigma_\back,\BAlgs)$ is compact, then
$\HSP_\Base$ and $\RedSgi$
are
statically
refutationally complete for all sufficiently complete sets of clauses,
\ie,
if a set of clauses $N$ is sufficiently complete
and saturated \wrt\ $\HSP_\Base$ and $\RedSgi$,
and $N \models_\BAlgs \Box$,
then $\Box \in N$.
\end{theorem}
\begin{proof}
Let $N$ be a set of weakly abstracted clauses that is sufficiently complete,
and saturated \wrt~the hierarchic superposition calculus and $\RedSgi$
and does not contain $\Box$.
Consequently, the \IR{Close} rule is not applicable to $N$.
By compactness, this means that the set of all $\Sigma_\back$-clauses
in $\sgi(N)$ is satisfied by some
term-generated $\Sigma_\back$-interpretation $I \in \BAlgs$.
By Thm.~\ref{thm:na-saturation},
$N_I$ is saturated with respect to the standard superposition calculus.
Since $\Box \notin N_I$, the refutational completeness of standard superposition
implies that there is a $\Sigma$-model $I'$ of $N_I$.
Since $N$ is sufficiently complete,
we know that for every ground term $t'$ of a BG sort
there exists a BG term $t$ such that $t' \approx t$ is true in $I'$.
Consequently, for every ground instance of a clause in $N$
there exists an equivalent simple ground instance,
thus $I'$ is also a model of all ground instances of clauses in $N$.
To see that the restriction of $I'$ to $\Sigma_\back$ is isomorphic to $I$
and thus in $\BAlgs$,
note that $I'$ satisfies $\EADA$,
preventing confusion, and that $N$ is sufficiently complete,
preventing junk.
Since $I'$ satisfies $N$
and $I'\rstr{\Sigma_\back} \in \BAlgs$, we have $N \not\models_\BAlgs \Box$
\qed
\end{proof}

A theorem proving derivation $\calD$ is a finite or infinite sequence of
weakly abstracted
clause sets $N_0, N_1, ...$, such that
$N_i$ and $N_{i+1}$ are equisatisfiable \wrt\ $\models_\BAlgs$
and
$N_i \setminus N_{i+1} \subseteq \RInfSgi(N_{i+1})$ for all indices $i$.
The set $N_\infty = \bigcup_{i\ge 0} \bigcap_{j\ge i} N_j$
is called the limit of $\calD$;
the set $N^\infty = \bigcup_{i\ge 0} N_i$ is called the
union of $\calD$.
It is easy to show that every clause in $N^\infty$
is either contained in $N_\infty$ or redundant \wrt~$N_\infty$.
The derivation $\calD$ is said to be fair, if
every $\HSP_\Base$-inference with (non-redundant) premises in $N_\infty$
becomes redundant at some point of the derivation.
The limit of a fair derivation is
saturated~\cite{Bachmair:Ganzinger:ResolutionTheoremProving:Handbook:2001};
this is the key result that allows us to deduce
dynamic refutational completeness
from
static refutational completeness:

\begin{theorem}
\label{thm:hiersup-dyn-ref-complete}%
If the BG specification $(\Sigma_\back,\BAlgs)$ is compact, then
$\HSP_\Base$ and $\RedSgi$
are
dynamically
refutationally complete for all sufficiently complete sets of clauses,
\ie,
if $N \models_\BAlgs \Box$,
then every fair derivation starting from $\abstr(N)$
eventually generates $\Box$.
\end{theorem}


In the rest of the paper, we consider only theorem proving derivations
where each set $N_{i+1}$ results from from $N_i$ by either adding
the conclusions of inferences from $N_i$,
or by deleting clauses that are redundant \wrt~$N_{i+1}$,
or by applying the following generic simplification rule for clause sets:
\begin{equation*}
\infrule[\IR{Simp}]{
N \cup \{ C \}
}{
N \cup \{ D_1,\ldots,D_n \}
} 
\end{equation*}
if $n \ge 0$ and
(i) $D_i$ is weakly abstracted, for all $i=1,\ldots,n$,
(ii) $N \cup \{ C \} \models_\BAlgs D_i$, and
(iii) $C \in \RClSgi(N \cup \{ D_1,\ldots,D_n \})$.

Condition (ii) is needed for soundness, and condition (iii) is needed for completeness.
The \IR{Simp} rule covers the usual simplification rules of the standard
superposition calculus, such as demodulation by unit clauses and deletion of
tautologies and (properly) subsumed clauses. It also covers simplification of
arithmetic terms, \eg, replacing a subterm $(2+3)+\alpha$ by $5+\alpha$ and deleting
an unsatisfiable BG literal $5+\alpha < 4+\alpha$ from a clause. 
Any clause of the form $C \vee \zeta  \not\approx d$ where $d$ is domain element can be
simplified to $C[\zeta \mapsto d]$. 
Notice, though, that impure BG terms or FG terms can in general not be
simplified by BG tautologies. Although, \eg, $\op f(X) + 1 \not\approx y + 1$ is larger than 
$1 + \op f(X) \not\approx y + 1$ (with a LPO), such a ``simplification'' is not justified by the
redundancy criterion. Indeed, in the example it destroys sufficient completeness. 
\bigskip

We have to point out a limitation of the calculus described
above.
The standard superposition calculus $\SSP$ exists in two variants:
either using the \IR{Equality factoring} rule,
or using the \IR{Factoring} and \IR{Merging paramodulation} rules.
Only the first of these variants works together with
weak abstraction.
Consider the following example.
Let $N = \{\,\alpha + \beta \approx \alpha,$
$\op c \not\approx \beta \lor \op c \not\approx 0,$
$\op c \approx \beta \lor \op c \approx 0\,\}$.
All clauses in $N$ are weakly abstracted.
Since the first clause entails $\beta \approx 0$ relative to linear arithmetic,
the second and the third clause are obviously contradictory.
The $\HSP_\Base$ calculus as defined above is able to detect this
by first applying \IR{Equality factoring} to the third clause,
yielding $\op c \approx 0 \lor \beta \not\approx 0$,
followed by two \IR{Negative superposition} steps and
\IR{Close}.
If \IR{Equality factoring} is replaced by
\IR{Factoring} and \IR{Merging paramodulation}, however,
the refutational completeness of $\HSP_\Base$ is lost.
The only inference that remains possible is
a \IR{Negative superposition} inference between the
third and the second clause.
But since the conclusion of this inference is a tautology,
the inference is redundant, so the clause set is saturated.
(Note that the clause $\beta \approx 0$ is entailed by $N$, but it
is not explicitly present, so there is no way to perform a
\IR{Merging paramodulation} inference with the smaller side of
the maximal literal of the third clause.)

\section{Local Sufficient Completeness}
\label{sec:local-sc}

The definition of sufficient completeness \wrt\ simple instances
that was given in Sect.~\ref{sec:hierarchic-tp}
requires that
\emph{every} ground BG-sorted
FG term $s$ is equal to some ground BG term $t$
in every $\Sigma$-model $J$ of $\sgi(N) \cup \GndTh(\BAlgs)$.
It is rather evident, however, that this condition is sometimes stronger than needed.
For instance, if the set of input clauses $N$ is ground,
then we only have to consider the ground BG-sorted
FG terms that actually occur in $N$~\cite{Kruglov:Weidenbach:MACIS:2012}
(analogously to the Nelson--Oppen
combination procedure).
A relaxation of sufficient completeness
that is also useful for non-ground clauses
and that still ensures refutational completeness
was given
by Kruglov~\cite{Kruglov2013PhD}:

\begin{definition}[Smooth ground instance]
\label{def:very-simple}
\label{def:smooth}
We say that a substitution $\sigma$ is \emph{smooth}
if for every variable $\zeta \in \dom(\sigma)$
all BG-sorted (proper or non-proper)
subterms of $\zeta\sigma$ are pure BG terms.
If $F\sigma$ is a ground instance of a term or clause $F$
and $\sigma$ is smooth,
$F\sigma$ is called a smooth ground instance.
(Recall that every ground BG term is necessarily pure.)
If $N$ is a set of clauses,
$\smgi(N)$ denotes the set of all smooth ground instances of clauses in $N$.
\qed
\end{definition}

Every smooth substitution is a simple substitution,
but not vice versa.
For instance, if $x$ is a FG-sorted variable and $y$ is
an ordinary BG-sorted variable,
then $\sigma_1 = [x \mapsto \op{cons}(\op{f}(1)+2,\op{empty})]$
and $\sigma_2 = [y \mapsto \op{f}(1)]$ are simple substitutions,
but neither of them is smooth,
since $x\sigma_1$ and $y\sigma_2$ contain the BG-sorted FG subterm $\op{f}(1)$.

\begin{definition}[Local sufficient completeness]
\label{def:local-sufficient-completeness}
Let $N$ be a $\Sigma$-clause set.
We say that $N$ is \emph{locally sufficiently complete \wrt\ smooth instances}
if for every $\Sigma_\back$-interpretation $I \in \BAlgs$,
every $\Sigma$-model $J$ of $\sgi(N) \cup \EADA$,
and every BG-sorted FG term $s$
occurring in $\smgi(N) \setminus \RClFlt(\smgi(N) \cup \EADA)$
there is
a ground BG term $t$ such that $J \models s \approx t$.
(Again, we will from now on omit the phrase ``\wrt~smooth instances'' for brevity.)
\qed
\end{definition}

\begin{example}
The clause set
$N = \{\,X \not\approx \alpha \lor \op f(X) \approx \beta\,\}$
is locally sufficiently complete:
The smooth ground instances
have the form
$s' \not\approx \alpha \lor \op f(s') \approx \beta$,
where $s'$ is a pure BG term.
We have to show that $\op f(s')$ equals some ground BG term $t$
whenever the smooth ground instance is not redundant.
Let $I \in \BAlgs$ be a $\Sigma_\back$-interpretation
and $J$ be a $\Sigma$-model of $\sgi(N) \cup \EADA$.
If $I \models s' \not\approx \alpha$,
then $s' \not\approx \alpha$ follows from some clauses in $\EADA$,
so $s' \not\approx \alpha \lor \op f(s') \approx \beta$
is contained in $\RClFlt(\smgi(N) \cup \EADA)$
and $\op f(s')$ need not be considered.
Otherwise $I \models s' \approx \alpha$,
then $\op f(s')$ occurs in a non-redundant
smooth ground instance of a clause in $N$
and $J \models f(s') \approx \beta$,
so $t := \beta$ has the desired property.
On the other hand,
$N$ is clearly not sufficiently complete,
since there are models of $\sgi(N) \cup \GndTh(\BAlgs)$
in which $\op f(\beta)$ is interpreted by some junk element
that is different from the interpretation of any ground BG term.
\end{example}

The example demonstrates that local sufficient completeness is
significantly more powerful than sufficient completeness,
but this comes at a price.
For instance,
as shown by the next example,
local sufficient completeness is not
preserved by abstraction:
\begin{example}
Suppose that the BG specification is linear integer arithmetic
(including parameters $\alpha$, $\beta$, $\gamma$),
the FG operators are
$\op f : \mathit{int} \to \mathit{int}$,
$\op g : \mathit{int} \to \mathit{data}$,
$\op a : \, \to \mathit{data}$,
the term ordering is an LPO with precedence
$\op g > \op f > \op a > \gamma > \beta > \alpha > 3 > 2 > 1$,
and the clause set $N$ is given by
\[
\begin{array}{@{}l@{\qquad}c@{}}
  \gamma \approx 1
&
  (1)
\\[1ex]
  \beta \approx 2
&
  (2)
\\[1ex]
  \alpha \approx 3
&
  (3)
\\[1ex]
  \op f(2) \approx 2
&
  (4)
\\[1ex]
  \op f(3) \approx 3
&
  (5)
\\[1ex]
  \op g(\op f(\alpha)) \approx \op a \;\lor\; \op g(\op f(\beta)) \approx \op a
&
  (6)
\\[1ex]
  \op g(\op f(\alpha)) \not\approx \op a \;\lor\; \op g(\op f(\beta)) \approx \op a
&
  (7)
\\[1ex]
  \op g(\op f(\gamma)) \approx \op a \;\lor\; \op g(\op f(\beta)) \approx \op a
&
  (8)
\end{array}\]
Since all clauses in $N$ are ground, $\smgi(N) = \sgi(N) = N$.
Clause (8) is redundant \wrt\ $\smgi(N) \cup \EADA$ (for any $I$):
it follows from clauses (6) and (7), and both are smaller than (8).
The BG-sorted FG terms in non-redundant clauses are
$\op f(2)$, $\op f(3)$, $\op f(\alpha)$, and $\op f(\beta)$,
and in any $\Sigma$-model $J$ of
$\sgi(N) \cup \EADA$, these are necessarily equal to the BG terms
$2$ or $3$, respectively,
so $N$ is locally sufficiently complete.

Let $N' = \abstr(N)$,
let $I$ be a BG-model such that
$\EA$ contains $\alpha \approx 3$, $\beta \approx 2$, and $\gamma \approx 1$
(among others),
$\DA$ contains $1 \not\approx 2$, $1 \not\approx 3$, and $2 \not\approx 3$
(among others),
and let $J$ be a $\Sigma$-model of
$\sgi(N') \cup \EADA$
in which $\op f(1)$ is interpreted by some junk element.
The set
$N'$ contains the clause
$\op g(\op f(X)) \approx \op a \;\lor\; \op g(\op f(Y)) \approx \op a
\;\lor\; \gamma \not\approx X \;\lor\; \beta \not\approx Y$
obtained by abstraction of (8).
Its smooth ground instance
$C \ =\ \op g(\op f(1)) \approx \op a \;\lor\; \op g(\op f(2)) \approx \op a
\;\lor\; \gamma \not\approx 1 \;\lor\; \beta \not\approx 2$
is not redundant:
it follows from other clauses in
$\smgi(N') \cup \EADA$,
namely
\[
\begin{array}{@{}l@{\qquad}l@{}}
  \alpha \approx 3
&
  (3)
\\[1ex]
  \op g(\op f(3)) \approx \op a \;\lor\; \op g(\op f(2)) \approx \op a \;\lor\; \alpha \not\approx 3 \;\lor\; \beta \not\approx 2
&
  (6')
\\[1ex]
  \op g(\op f(3)) \not\approx \op a \;\lor\; \op g(\op f(2)) \approx \op a \;\lor\; \alpha \not\approx 3 \;\lor\; \beta \not\approx 2
&
  (7')
\end{array}\]
but the ground instances $(6')$ and $(7')$
that are needed here are larger than $C$.
Since $C$ contains the BG-sorted FG term $\op f(1)$
which is interpreted differently from any BG term in $J$,
$N'$ is not locally sufficiently complete.
\end{example}

Local sufficient completeness
of a clause set suffices to ensure refutational completeness.
Kruglov's proof~\cite{Kruglov2013PhD}
works also if one uses
weak abstraction instead of strong abstraction
and ordinary as well as abstraction variables,
but it relies on an additional restriction on the term ordering.\footnote{%
  In~\cite{Kruglov2013PhD}, it is required that every ground term
  that contains a (proper or improper)
  BG-sorted FG subterm must be larger than
  any (BG or FG) ground term that does not contain such a subterm.}
We give an alternative proof that works without this restriction.

The proof is based on a transformation on $\Sigma$-interpretations.
Let $J$ be an arbitrary $\Sigma$-interpretation.
We transform $J$ into a term-generated $\Sigma$-interpretation
$\nojunk(J)$ without junk in two steps.
In the first step, we define
a $\Sigma$-inter\-pre\-ta\-tion $J'$ as follows:
\begin{itemize}
\item
  For every FG sort $\xi$,
  define $J'_\xi = J_\xi$.
\item
  For every BG sort $\xi$,
  define $J'_\xi = \{\,t^J \mid \text{$t$~is a ground BG term of sort~$\xi$}\,\}$.
\item
  For every $f : \xi_1 \ldots \xi_n \to \xi_0$
  the function $J'_f : J'_{\xi_1} \times \cdots \times J'_{\xi_n} \to J'_{\xi_0}$
  maps $(a_1,\dots,a_n)$ to $J_f(a_1,\dots,a_n)$, if
  $J_f(a_1,\dots,a_n) \in J'_{\xi_0}$,
  and to an arbitrary element of $J'_{\xi_0}$ otherwise.
\end{itemize}
That is, we obtain $J'$ from $J$ be deleting all junk elements from
$J_\xi$ if $\xi$ is a BG sort,
and by redefining the interpretation of $f$ arbitrarily whenever
$J_f(a_1,\dots,a_n)$ is a junk element.

In the second step, we define the $\Sigma$-interpretation
$\nojunk(J) = J''$ as the term-generated subinterpretation of $J'$, that is,
\begin{itemize}
\item
  For every sort $\xi$,
  $J''_\xi = \{\,t^{J'} \mid \text{$t$~is a ground term of sort $\xi$}\,\}$,
\item
  For every $f : \xi_1 \ldots \xi_n \to \xi_0$,
  the function
  $J''_f : J''_{\xi_1} \times \cdots \times J''_{\xi_n} \to J''_{\xi_0}$
  satisfies $J''_f(a_1,\dots,a_n) = J'_f(a_1,\dots,a_n)$.
\end{itemize}

\begin{lemma}
\label{lem:vsterms-keep-int}
Let $J$, $J'$, and $\nojunk(J) = J''$ be given as above.
Then the following properties hold:
\begin{enumerate}[\rm(i)]
\item
  $t^{J''} = t^{J'}$ for every ground term $t$.
\item
  $J''_\xi = J'_\xi$ for every BG sort $\xi$.
\item
  $J''$ is a term-generated $\Sigma$-interpretation
  and $J''\rstr{\Sigma_\back}$ is a term-generated $\Sigma_\back$-inter\-pre\-ta\-tion.
\item
  If $t = f(t_1,\dots,t_n)$ is a ground term,
  $t_i^{J'} = t_i^J$ for all $i$, and $t^J \in J'_\xi$,
  then $t^{J'} = t^J$.
\item
  If $t$ is a ground term such that
  all BG-sorted subterms of $t$ are BG terms, then $t^{J'} = t^J$.
\item
  If $C$ is a ground BG clause, then $J \models C$ if and only if $J'' \models C$
  if and only if $J''\rstr{\Sigma_\back} \models C$.
\end{enumerate}
\end{lemma}
\begin{proof}
Properties (i)-(iv) follow directly from the definition of $J'$ and $J''$.
Property (v) follows from (iv) and
the definition of $J'$ by induction over the term structure.
By (i) and (v), every ground BG term is interpreted in the same way in $J$ and $J''$,
moreover it is obvious that every ground BG term is interpreted in the same way
in $J''$ and $J''\rstr{\Sigma_\back}$;
this implies~(vi).
\qed
\end{proof}

\begin{lemma}
If $J$ is a $\Sigma$-interpretation and
$I = \nojunk(J)$,
then for every ground term $s$ there exists a ground term $t$
such that $s^I = t^I$ and all BG-sorted (proper or non-proper)
subterms of $t$ are BG terms.
\end{lemma}
\begin{proof}
If $s$ has a BG sort $\xi$, then this follows directly from the
fact that $s^I \in I_\xi$ and that
every element of $I_\xi$ equals $t^I$ for some ground BG term $t$ of sort~$\xi$.
If $s$ has a FG sort, let $s_1,\dots,s_n$ be the maximal BG-sorted subterms
of $s = s[s_1,\dots,s_n]$.
Since for every $s_i$ there is a ground BG term $t_i$
with $s_i^I = t_i^I$,
we obtain $s^I = (s[s_1,\dots,s_n])^I = (s[t_1,\dots,t_n])^I$.
Set $t := s[t_1,\dots,t_n]$.
\qed
\end{proof}

\begin{corollary}
\label{cor:subst-equiv-vs-subst}
Let $J$ be a $\Sigma$-interpretation and
$I = \nojunk(J)$.
Let $C\sigma$ by a ground instance of a clause $C$.
Then there is a smooth ground instance $C\tau$ of $C$
such that $(t\sigma)^I = (t\tau)^I$ for every term occurring in $C$
and such that $I \models C\sigma$ if and only if $I \models C\tau$.
\end{corollary}
\begin{proof}
Using the previous lemma, we
define $\tau$ such that for every variable $\zeta$ occurring in $C$,
$(\zeta\tau)^I = (\zeta\sigma)^I$ and all BG-sorted (proper or non-proper)
subterms of $\zeta\tau$ are BG terms.
Clearly $\tau$ is smooth.
The other properties follow immediately
by induction over the term or clause structure.
\qed
\end{proof}

\begin{lemma}
\label{lem:interp-equiv-nojunk-interp}
Let $N$ be a set of $\Sigma$-clauses
that is locally sufficiently complete.
Let $I \in \BAlgs$ be a $\Sigma_\back$-interpretation,
let $J$ be a $\Sigma$-model of $\sgi(N) \cup \EADA$,
and let $J'' = \nojunk(J)$.
Let $C \in N$ and
let $C\tau$ by a smooth ground instance in
$\smgi(N) \setminus \RClFlt(\smgi(N) \cup \EADA)$.
Then
$(t\tau)^J = (t\tau)^{J''}$ for every term $t$ occurring in $C$
and $J \models C\tau$ if and only if $J'' \models C\tau$.
\end{lemma}
\begin{proof}
Let $J'$ be defined as above, then $(t\tau)^{J'} = (t\tau)^{J''}$
for any term $t$ occurring in $C$
by Lemma~\ref{lem:vsterms-keep-int}-(i).
We prove that $(t\tau)^J = (t\tau)^{J'}$ by induction over the term structure:
If $t$ is a variable, then by smoothness
all BG-sorted subterms of $t\tau$
are BG terms, hence $(t\tau)^{J'} = (t\tau)^J$ by Lemma~\ref{lem:vsterms-keep-int}-(v).
Otherwise let $t = f(t_1,\dots,t_n)$.
If $t\tau$ is a BG term,
then again $(t\tau)^{J'} = (t\tau)^J$ by Lemma~\ref{lem:vsterms-keep-int}-(v).
If $t\tau$ is a FG term of sort $\xi$, then $t$ must be a FG term of sort $\xi$ as well.
By the induction hypothesis,
$(t_i\tau)^J = (t_i\tau)^{J'}$ for every $i$.
If $\xi$ is a FG sort,
then trivially $(t\tau)^J = J_f((t_1\tau)^J,\dots,(t_n\tau)^J)$
is contained in $J'_\xi$, so $(t\tau)^{J'} = (t\tau)^J$
by Lemma~\ref{lem:vsterms-keep-int}-(iv).
Otherwise, $t\tau$ is a
BG-sorted FG term
occurring in $\smgi(N) \setminus \RClFlt(\smgi(N) \cup \EADA)$.
By local sufficient completeness, there exists a ground BG term $s$
such that $s^J = (t\tau)^J$, hence $(t\tau)^J \in J'_\xi$.
Again, Lemma~\ref{lem:vsterms-keep-int}-(iv) yields $(t\tau)^{J'} = (t\tau)^J$.

Since all left and right-hand sides of equations in $C\tau$
are evaluated in the same way in $J''$ and $J$,
it follows that $J \models C\tau$ if and only if $J'' \models C\tau$.
\qed
\end{proof}

\begin{lemma}
\label{lem:hspbase-ref-compl-lsc-aux}%
Let $N$ be a set of $\Sigma$-clauses
that is locally sufficiently complete.
Let $I \in \BAlgs$ be a $\Sigma_\back$-interpretation,
let $J$ be a $\Sigma$-model of $\sgi(N) \cup \EADA$,
and let $J'' = \nojunk(J)$.
Then $J''$ is a model of $N$.
\end{lemma}
\begin{proof}
The proof proceeds in three steps.
In the first step
we show that
$J''$ is a model of $\smgi(N) \setminus \RClFlt(\smgi(N) \cup \EADA)$:
Let $C \in N$
and let
$C\tau$ be a smooth ground instance
in $\smgi(N) \setminus \RClFlt(\smgi(N) \cup \EADA)$.
Since every smooth ground instance is a simple ground instance
and $J$ is a $\Sigma$-model of $\sgi(N)$,
we know that $J \models C\tau$.
By Lemma~\ref{lem:interp-equiv-nojunk-interp},
this implies $J'' \models C\tau$.

In the second step we show
that $J''$ is a model of $\smgi(N)$.
Since we already know that
$J''$ is a model of $\smgi(N) \setminus \RClFlt(\smgi(N) \cup \EADA)$,
it is clearly sufficient to show that
$J''$ is a model of $\RClFlt(\smgi(N) \cup \EADA)$:
First we observe that
by Lemma~\ref{lem:vsterms-keep-int}
$J'' \models \EADA$.
Using the result of the first step,
this implies that
$J''$ is a model of
$(\smgi(N) \setminus \RClFlt(\smgi(N) \cup \EADA)) \cup \EADA$,
and since
this set is a superset of
$(\smgi(N) \cup \EADA) \setminus \RClFlt(\smgi(N) \cup \EADA)$,
$J''$ is also a model of the latter.
By Def.~\ref{dfn:redundancy-criterion}-(i),
$(\smgi(N) \cup \EADA) \setminus \RClFlt(\smgi(N) \cup \EADA)
\models \RClFlt(\smgi(N) \cup \EADA)$.
So $J''$ is a model
of $\RClFlt(\smgi(N) \cup \EADA)$.

We can now show the main result of the lemma:
We know that $J''$ is a term-generated $\Sigma$-interpretation,
so $J'' \models N$ holds if and only if $J''$ is a model of all ground instances of
clauses in $N$.
Let $C\sigma$ be an arbitrary
ground instance of $C \in N$.
By Cor.~\ref{cor:subst-equiv-vs-subst},
there is a smooth ground instance $C\tau$ of $C$
such that $J'' \models C\sigma$ if and only if $J'' \models C\tau$.
As the latter has been shown in the second step,
\looseness=-1
the result follows.
\qed
\end{proof}

\begin{theorem}
\label{thm:hspbase-ref-compl-lsc}%
If the BG specification $(\Sigma_\back,\BAlgs)$ is compact and
if the clause set $N$ is locally sufficiently complete,
then $\HSP_\Base$ and $\RedSgi$
are
dynamically refutationally complete for $\abstr(N)$,
\ie,
if $N \models_\BAlgs \Box$,
then every fair derivation starting from $\abstr(N)$
eventually generates $\Box$.
\end{theorem}
\begin{proof}
Let $\calD = (N_i)_{i\ge 0}$ be a fair derivation starting from $N_0 = \abstr(N)$,
and let $N_\infty$ be the limit of $\calD$.
By fairness, $N_\infty$ is saturated \wrt~$\HSP_\Base$ and $\RedSgi$.
If $\Box \notin N_\infty$,
then the \IR{Close} rule is not applicable to $N_\infty$.
Since $(\Sigma_\back,\BAlgs)$ is compact,
this means that the set of all $\Sigma_\back$-clauses
in $\sgi(N_\infty)$ is satisfied by some
term-generated $\Sigma_\back$-interpretation $I \in \BAlgs$.
By Thm.~\ref{thm:na-saturation},
$(N_\infty)_I$ is saturated with respect to the standard superposition calculus.
Since $\Box \notin (N_\infty)_I$, the refutational completeness of standard superposition
implies that there is a $\Sigma$-model $J$ of $(N_\infty)_I$,
and since $\EADA \subseteq (N_\infty)_I$,
$J$ is also a $\Sigma$-model of $\sgi(N_\infty) \cup \EADA$.
Since every clause in $N_0$ is either contained
in $N_\infty$ or redundant \wrt~$N_\infty$,
every simple ground instance of a clause in $N_0$
is a simple ground instance of a clause in $N_\infty$
or contained in $\GndTh(\BAlgs)$
or redundant \wrt~$\sgi(N_\infty)\cup\GndTh(\BAlgs)$.
We conclude that $J$ is a $\Sigma$-model of $\sgi(N_0)$,
and since $\sgi(N_0)$ and $\sgi(N)$ are equivalent,
$J$ is a $\Sigma$-model of $\sgi(N)$.
Now define $J'' = \nojunk(J)$.
By Lemma~\ref{lem:vsterms-keep-int},
$J''$ is a term-generated $\Sigma$-interpretation,
$J''\rstr{\Sigma_\back}$ is a term-generated $\Sigma_\back$-interpretation,
and $J''\rstr{\Sigma_\back}$ satisfies $\EADA$.
Consequently, $J''\rstr{\Sigma_\back}$ is isomorphic to $I$
and thus contained in $\BAlgs$.
Finally,
$J''$ is a model of $N$
by Lemma~\ref{lem:hspbase-ref-compl-lsc-aux}.
\qed
\end{proof}

If all BG-sorted FG terms in a set $N$ of clauses are ground,
local sufficient completeness can be established automatically
by adding a ``definition'' of the form $t \approx \alpha$,
where $t$ is a ground BG-sorted FG term and $\alpha$ is a parameter.
The following section explains this idea in a more general way.

\section[Local Sufficient Completeness by {\sf Define}]{Local Sufficient Completeness by Define}
\label{sec:define}
The $\HSP_\Base$ inference system will derive a contradiction if the
input clause set is inconsistent and (locally) sufficiently complete (cf.\ Sect.~\ref{sec:local-sc}). 
In this section we extend this functionality by adding an inference rule, \IR{Define},  which can
turn input clause sets that are not sufficiently complete into locally sufficiently
complete ones.
Technically, the \IR{Define} rule
derives ``definitions'' of the form $t \approx \alpha$, where $t$ is a ground BG-sorted
FG term and $\alpha$ is a parameter of the proper sort.
For economy of 
reasoning, definitions are introduced only on a by-need basis, when $t$ appears in a
current clause, and $t \approx \alpha$ is used to simplify that clause immediately. 

We need one more preliminary definition before introducing \IR{Define} formally.
\begin{definition}[Unabstracted clause]
A clause is \emph{unabstracted} if it does not contain any disequation $\zeta \not\approx t$
between a variable $\zeta$ and a term $t$ unless 
$t \neq \zeta$ and $\zeta \in \vars(t)$.  
\qed
\end{definition}
Any clause can be unabstracted by repeatedly replacing $C \vee \zeta \not\approx t$ by
${C[\zeta \mapsto t]}$ whenever $t = \zeta$ or $\zeta \notin \vars(t)$.
Let \emph{$\unabstr(C)$} denote an unabstracted version
of $C$ obtained this way.
If $t = t[\zeta_1,\ldots,\zeta_n]$ is a term in $C$ and $\zeta_i$ is
finally instantiated to $t_i$, we denote its unabstracted version $t[t_1,\ldots,t_n]$ by
$\unabstr(t[\zeta_1,\ldots,\zeta_n], C)$.
For a clause set $N$ let $\unabstr(N) = \{\, \unabstr(C) \mid C \in N \,\}$.

The \emph{full inference system $\HSP$}
of the hierarchic superposition calculus
consists of the inference rules of $\HSP_\Base$ and the
following \IR{Define} inference rule. As for the other inference rules we suppose
that all premises are weakly abstracted.
\COMMENT{PB}{4/3/2019}{
Applied abstr-operator to the set of the  conclusion clauses, in the inference rule
definition and condition (v).  Before we had the set of the abstr-operator
applied to these two clauses separately. A bit shorter and avoids an overfull line in condition (v)
}
\begin{equation*}
\infrule[\IR{Define}]{
N \cup \{ L[t[\zeta_1,\ldots,\zeta_n]] \vee D \}
}{
N \cup \abstr(\{ t[t_1,\ldots,t_n] \approx \alpha_{t[t_1,\ldots,t_n]},\ L[\alpha_{t[t_1,\ldots,t_n]}] \vee D\}
} 
\end{equation*}
if
\begin{enumerate}[\rm(i)]
\item $t[\zeta_1,\ldots,\zeta_n]$ is a minimal BG-sorted non-variable term with a
toplevel FG operator, 
\item $t[t_1,\ldots,t_n] = \unabstr(\{t[\zeta_1,\ldots,\zeta_n], L[t[\zeta_1,\ldots,\zeta_n]] \vee D\})$, 
\item $t[t_1,\ldots,t_n]$ is ground, 
\item $\alpha_{t[t_1,\ldots,t_n]}$ is a parameter, uniquely determined by  $t[t_1,\ldots,t_n]$, and
\item $L[t[\zeta_1,\ldots,\zeta_n]] \vee D \in \RClSgi(N \cup \abstr(\{ t[t_1,\ldots,t_n] \approx \alpha_{t[t_1,\ldots,t_n]},\ L[\alpha_{t[t_1,\ldots,t_n]}] \vee D\}))$.
\end{enumerate}
\bigskip

In (i), by minimality we mean that no proper subterm of $t[\zeta_1,\ldots,\zeta_n]$ is a 
BG-sorted non-variable term with a toplevel FG operator. 
In effect,
the \IR{Define} rule eliminates such terms inside-out. 
Conditions (iii) and (iv) are needed for soundness. Condition (v) 
is needed to
guarantee that \IR{Define} is a simplifying inference rule, much like the \IR{Simp} rule
in Sect.~\ref{sec:refutational-completeness}.\footnote{Condition (i) of \IR{Simp} is
  obviously satisfied and condition (iii) there is condition (v) of \IR{Define}. Instead of condition
  (ii), \IR{Define} inferences are only $\BAlgs$-satisfiability preserving, which however
  does not endanger soundness.} In particular, it makes sure that \IR{Define} cannot be applied to
definitions themselves.

\begin{theorem}
\label{thm:hsp-sat-preserving}%
The inference rules of $\HSP$
are satisfiability-preserving \wrt\ $\models_\BAlgs$,
\ie, for every inference with premise $N$ and conclusion $N'$
we have
$N \models_\BAlgs \Box$ if and only if $N' \models_\BAlgs \Box$.
Moreover,
$N' \models_\BAlgs N$.
\end{theorem}
\begin{proof}
For the inference rules of $\HSP_\Base$,
the result follows from Thm.~\ref{thm:hsp-base-sound}.

For \IR{Define}, we observe first that condition (ii) implies that
$L[t[\zeta_1,\ldots,\zeta_n]] \vee D$
and
$L[t[t_1,\ldots,t_n]] \vee D$
are equivalent.
If $N \cup \{ L[t[t_1,\ldots,t_n]] \vee D \}$
is $\BAlgs$-satisfiable,
let $I$ be a $\Sigma$-model of all ground instances of
$N \cup \{ L[t[t_1,\ldots,t_n]] \vee D \}$ such that
$I\rstr{\Sigma_\back}$ is in $\BAlgs$.
By condition (iii), $t[t_1,\ldots,t_n]$ is ground.
Let $J$ be the $\Sigma$-interpretation obtained from $J$
by redefining the interpretation of $\alpha_{t[t_1,\ldots,t_n]}$
in such a way that
$\alpha_{t[t_1,\ldots,t_n]}^J = t[t_1,\ldots,t_n]^I$,
then $J$ is a $\Sigma$-model of every ground instance of $N$,
$t[t_1,\ldots,t_n] \approx \alpha_{t[t_1,\ldots,t_n]}$
and $L[\alpha_{t[t_1,\ldots,t_n]}] \vee D$,
and hence also a model of the abstractions of these clauses.
Conversely,
every model of
$t[t_1,\ldots,t_n] \approx \alpha_{t[t_1,\ldots,t_n]}$
and
$L[\alpha_{t[t_1,\ldots,t_n]}] \vee D$
is a model of
$L[t[t_1,\ldots,t_n]] \vee D$.
\qed
\end{proof}


\begin{example}
  \label{ex:define-1}
Let $C = \op g(\op f(x, y)+1, x, y) \approx 1 \vee x \not\approx 1 + \beta \vee y \not\approx \op c$ be the premise of a
$\IR{Define}$ inference. We get 
$\unabstr(C) =  \op g(\op f(1+\beta, \op c)+1, 1+\beta, \op c) \approx 1$. The (unabstracted)
conclusions are the definition $\op f(1+\beta, \op c) \approx \alpha_{\op f(1+\beta, \op c)}$ and the
clause $\op g(\alpha_{\op f(1+\beta, \op c)}+1, x, y) \approx 1 \vee x \not\approx 1 + \beta \vee y \not\approx \op c$.
Abstraction yields
$\op f(X, \op c) \approx \alpha_{\op f(1+\beta, \op c)}
\vee
X \not\approx 1+\beta$
and
$\op g(Z, x, y) \approx 1 \vee x \not\approx 1 + \beta \vee y \not\approx \op c
\vee
Z \not\approx \alpha_{\op f(1+\beta, \op c)}+1$.

One might be tempted to first unabstract the premise $C$ before applying $\IR{Define}$.
However, unabstraction may eliminate FG terms ($\op c$ in the example)
which is not undone by abstraction. This may lead to incompleteness.\qed
\end{example}

\begin{example}
  \label{ex:define-2}
  The following clause set demonstrates the need for condition (v) in \IR{Define}.
  Let $N = \{ \op f(\op c) \approx 1\}$ and suppose condition (v) is absent. Then we obtain 
$N' = \{ \op f(\op c) \approx \alpha_{\op f(\op c)},\ \alpha_{\op f(\op c)} \approx 1\}$. By demodulating
the first clause with the second clause we get 
$N'' = \{ \op f(\op c) \approx 1,\ \alpha_{\op  f(\op c)} \approx 1\}$. 
Now we can continue with $N''$ as with $N$. The problem is, of course, that the new definition
$\op f(\op c) \approx \alpha_{\op f(\op c)}$ is greater \wrt\ the term ordering than the parent
clause, in violation of condition~(v).
\qed
\end{example}

\begin{example} 
\label{ex:define-3}
Consider the weakly abstracted clauses $\op P(0)$, $\op f(x)>0 \lor  \lnot \op P(x)$,    
$\op Q(\op f(x))$, $\lnot \op Q(x) \lor  0>x$.
  Suppose $\neg \op P(x)$ is maximal in the second clause. 
By superposition between the first two clauses we derive $\op f(0)>0$.
With \IR{Define} we obtain $\op f(0)\approx \alpha_{\op f(0)}$ and $\alpha_{\op f(0)}>0$, the latter
replacing $\op 
f(0)>0$. From the third clause and $\op f(0)\approx \alpha_{\op f(0)}$ we obtain $\op Q(\alpha_{\op f(0)})$, and with
the fourth clause $0>\alpha_{\op f(0)}$. Finally we apply \IR{Close} to 
$\{ \alpha_{\op f(0)}>0,\ 0>\alpha_{\op f(0)}\}$. \qed
\end{example}

\COMMENT{PB}{28/09/2019}{
Removed, as not important:

It can happen that one unabstracted version of a clause $C$
has a BG-sorted
ground term, thus permitting \IR{Define}, and another one does not,
say, if $C = X \not\approx Y+1 \lor X \not\approx 6 \lor f(X) \approx 1$.
In practice, one could use a heuristics
that prefers those unabstractions that 
enable \IR{Define}. However, in our main application, the ground BG-sorted term
fragment in Sect.~\ref{sec:gbt},
unabstraction is always unique and so this problem does not arise.
}
It is easy to generalize Thm.~\ref{thm:hspbase-ref-compl-lsc}
to the case that local sufficient completeness does not hold initially,
but is only established with the help of \IR{Define} inferences:

\begin{theorem}
\label{thm:hspbase-ref-compl-lsc-define}%
Let $\calD = (N_i)_{i\ge 0}$ be a fair $\HSP$ derivation starting from
$N_0 = \abstr(N)$, let $k \ge 0$,
such that $N_k = \abstr(N')$ and $N'$ is
locally sufficiently complete.
If the BG specification $(\Sigma_\back,\BAlgs)$ is compact,
then the limit of $\calD$ contains $\Box$ if and only if $N$ is
$\BAlgs$-unsatisfiable.
\end{theorem}
\begin{proof}
Since every derivation step in an $\HSP$ derivation is
satisfiability-preserving, the ``only if'' part is
again obvious.

For the ``if'' part, we assume that
$N_\infty$, the limit of $\calD$, does not contain $\Box$.
By fairness, $N_\infty$ is saturated \wrt~$\HSP$ and $\RedSgi$.
We start by considering the subderivation
$(N_i)_{i\ge k}$ starting with $N_k = \abstr(N')$.
Like in the proof of Thm.~\ref{thm:hspbase-ref-compl-lsc}, we
can show that $N'$ is $\BAlgs$-satisfiable,
that is, there exists a model $J$ of $N'$
that is a term-generated $\Sigma$-interpretation,
and whose restriction $J\rstr{\Sigma_\back}$ is contained
in $\BAlgs$.
From Lemma~\ref{lemma:modelssgi-implies-modelsbalgs}
and Prop.~\ref{prop:wab-equivalence-transformation}
we see that $N' \models_\BAlgs N_k$,
and similarly $N_0 \models_\BAlgs N$.
Furthermore, since every clause in $N_0 \setminus N_k$
must be redundant \wrt~$N_k$,
we have $N_k \models_\BAlgs N_0$.
Combining these three entailments,
we conclude that $N' \models_\BAlgs N$, so
$N$ is $\BAlgs$-satisfiable
and $J$ is a model of $N$.
\qed
\end{proof}


Condition (v) of the \IR{Define} rule
requires that the clause that is deleted during a \IR{Define} inference
must be redundant with respect to the remaining clauses. This condition is needed to
preserve refutational completeness. There are cases, however, where condition (v) prevents
us from introducing a definition for a subterm.
Consider the clause set $N =
\{ C \}$ where $C  = \op f (\op c) \approx 1 \vee \op c \approx \op d$,
the constants $\op c$ and $\op d$ are FG-sorted,
$\op f$ is a BG-sorted FG operator,
and $\op c \succ \op d \succ 1$.
The literal $\op f (\op c) \approx 1$ is maximal in $C$.
The clause set $N = \abstr(N)$ is not locally sufficient complete
(the BG-sorted FG-term $\op f (\op c)$ may be interpreted differently
from all BG terms in a $\Sigma$-model).
Moreover, it cannot be made locally sufficient complete
using the \IR{Define} rule,
since the definition $\op f (\op c) \approx \alpha_{\op f (\op c)}$
is larger \wrt\ the clause ordering than $C$, in violation of condition (v)
of \IR{Define}.

However,
at the beginning of a derivation, we may be a bit more permissive.
Let us
define the \emph{reckless \IR{Define}} inference rule in the same way
as the \IR{Define} rule except that the applicability condition (v) is dropped.
Clearly, in the example above,
the reckless \IR{Define} rule allows us to derive the locally
sufficiently complete clause  set
$N' = \{ \alpha_{\op f  (\op c)} \approx 1 \vee \op c \approx \op d,\ \op f (\op c) \approx \alpha_{\op f (\op c)}\}$ as
desired.
In fact, we can show that this is always possible
if $N$ is a finite clause set in which all BG-sorted FG terms are ground.

\begin{definition}[Pre-derivation]
Let $N_0$ be a weakly abstracted clause set. A
\emph{pre-derivation (of a clause set $N^\pre$)} is
a derivation of the form $N_0,N_1,$\,\linebreak[2]$\ldots,(N_k = N^\pre)$,
for some $k \ge 0$, with the
inference rule reckless \IR{Define} only, and such that each clause $C \in N^\pre$ either does not contain 
any BG-sorted FG operator or $C = \abstr(C')$ and $C'$ is a \emph{definition}, \ie, a ground positive unit 
clause of the form $\op f(t_1,\ldots,t_n) \approx t$ where $\op f$ is a BG-sorted
FG operator, $t_1,\ldots,t_n$ do not contain BG-sorted FG operators, and $t$ is a
background term. \qed
\end{definition}

\begin{lemma}
\label{lemma:prederiv-yields-lsc}
Let $N$ be a finite clause set in which all BG-sorted FG terms are ground.
Then there is a pre-derivation starting from $N_0 = \abstr(N)$
such that $N^\pre$ is locally sufficiently complete.
\end{lemma}
\begin{proof}
Since every term headed by a BG-sorted FG operator in
$\unabstr(N_0)$ is ground, we can incrementally eliminate
all occurrences of terms headed by BG-sorted FG operators from $N_0$,
except those in abstractions of definitions.
Let $N_0,N_1,\ldots,(N_k = N^\pre)$ be the sequence of sets of clauses
obtained in this way.
We will show that $N^\pre$ is locally sufficiently complete.

Let $I \in \BAlgs$ be a $\Sigma_\back$-interpretation,
let $J$ be a $\Sigma$-model of $\sgi(N^\pre) \cup \EADA$
and let $C\theta$ be a smooth ground instance
in $\smgi(N) \setminus \RClFlt(\smgi(N) \cup \EADA)$.
We have to show that
for every BG-sorted FG term $s$ occurring in $C\theta$
there is a ground BG term $t$ such that $J \models s \approx t$.

If $C$ does not contain any BG-sorted FG operator, then
there are no BG-sorted FG terms in $C\theta$,
so the property is vacuously true.
Otherwise $C = \abstr(C')$ and $C'$ is a definition
$\op f(t_1,\ldots,t_n) \approx t$ where $\op f$ is a BG-sorted
FG operator, $t_1,\ldots,t_n$ do not contain BG-sorted FG operators, and $t$ is a
background term.
In this case,
$C$ must have the form
$\op f(u_1,\ldots,u_n) \approx u \lor E$, such that
$E$ is a BG clause,
$u_1,\ldots,u_n$ do not contain BG-sorted FG operators, and $u$ is a
BG term.
The only BG-sorted FG term in the smooth instance $C\theta$ is
therefore
$\op f(u_1\theta,\ldots,u_n\theta)$.
If any literal of $E\theta$ were true in $J$, then it would
follow from $\EADA$, therefore
$C\theta \in \RClFlt(\smgi(N) \cup \EADA)$,
contradicting the assumption.
Hence $J \models \op f(u_1\theta,\ldots,u_n\theta) \approx u\theta$,
and since $u\theta$ is a ground BG term,
the requirement is satisfied.
\qed
\end{proof}
Lemma~\ref{lemma:prederiv-yields-lsc} will be needed to prove a completeness result for
the fragment defined in the next section.

\section{The Ground BG-sorted Term Fragment}
\label{sec:gbt}

According to Thm.~\ref{thm:hspbase-ref-compl-lsc},
the $\HSP_\Base$ calculus is refutationally complete
provided that the clause set
is locally sufficiently complete
and the BG specification is compact.
We have seen in the previous section that
the (reckless) \IR{Define} rule can help to
establish local sufficient completeness
by introducing new parameters.
In fact,
finite clause sets in which all BG-sorted FG terms are ground
can always be converted into
locally sufficiently complete clause sets (cf.\ Lemma~\ref{lemma:prederiv-yields-lsc}).
On the other hand,
as noticed in Sect.~\ref{sec:hierarchic-tp},
the introduction of
parameters can destroy the compactness of the BG specification.
In this and the following section,
we will identify two cases
where we can not only establish local sufficient completeness,
but where we can also guarantee that compactness poses no problems.
The \emph{ground BG-sorted term fragment (GBT fragment)} is one
such case:


\begin{definition}[GBT fragment]
A clause $C$ is a GBT clause if
all BG-sorted terms in $C$ are ground.
A finite clause set $N$ belongs to the GBT fragment if
all clauses in $N$ are GBT clauses.
\qed
\end{definition}

Clearly, by Lemma~\ref{lemma:prederiv-yields-lsc}
for every clause set $N$ that
belongs to the GBT fragment
there is a pre-derivation that converts $\abstr(N)$ into a
locally sufficiently complete clause set.
Moreover, pre-derivations also preserve the GBT property:


\begin{lemma}
\label{lemma:define-preserves-GBT}
 If $\unabstr(N)$ belongs
to the GBT fragment and $N'$ is obtained from $N$ by a reckless \IR{Define} inference, then
$\unabstr(N')$ also belongs to the GBT fragment. 
\end{lemma}
\COMMENT{PB}{1/03/2019}{
Proof of this lemma retained only for long version. I think it is safe to omit it.
}
  \begin{proof}
    Suppose that $\unabstr(N)$ belongs to the GBT fragment and let\linebreak[4]
    $L[t[\zeta_1,\ldots,\zeta_n]] \vee C$ be a clause in $N$ to which reckless \IR{Define} is applied.
    The reckless \IR{Define} inference results in
    the two conclusion clauses
    $C_1 = \abstr(t[t_1,\ldots,t_n] \approx \alpha_{t[t_1,\ldots,t_n]})$ and
    $C_2 = \abstr(L[\alpha_{t[t_1,\ldots,t_n]}] \vee D)$.  It suffices to show that
    $\unabstr(C_1)$ and $\unabstr(C_2)$ are GBT clauses.

    Because $\unabstr(N)$ belongs to the GBT fragment, every BG-sorted subterm in the
    unabstracted version $L[t[t_1,\ldots,t_n]] \vee D$ of
    $L[t[\zeta_1,\ldots,\zeta_n]] \vee C$ is ground.  It follows trivially that both
    $t[t_1,\ldots,t_n] \approx \alpha_{t[t_1,\ldots,t_n]}$ and
    $L[\alpha_{t[t_1,\ldots,t_n]}] \vee D$ belong to the GBT fragment as well. It is fairly easy to see
    that for GBT clauses unabstraction reverses abstraction. Therefore
    $\unabstr(C_1) = t[t_1,\ldots,t_n] \approx \alpha_{t[t_1,\ldots,t_n]}$ and
    $\unabstr(C_2) = L[\alpha_{t[t_1,\ldots,t_n]}] \vee D$. As just shown, both clauses belong to the
    GBT fragment.  \qed
  \end{proof}

As we have seen, $N^\pre$ is locally sufficiently complete.
At this stage this suggests to exploit the completeness result for
locally sufficiently complete clause sets,
Thm.~\ref{thm:hspbase-ref-compl-lsc}. However, Thm.~\ref{thm:hspbase-ref-compl-lsc}
requires compact BG specifications, and the question is if we can avoid this. We can indeed
get a complete calculus under rather mild assumptions on the \IR{Simp} rule:
\begin{definition}[Suitable \IR{Simp} inference]
\label{def:suitable-simp}
  Let  $\succ_{\mathrm{fin}}$ be a strict partial term ordering such that for every ground 
BG term $s$ only finitely many ground BG terms $t$ with $s
\succ_{\mathrm{fin}} t$ exist.\footnote{%
  A KBO with appropriate weights can be used for $\succ_{\mathrm{fin}}$.}
We say that a \IR{Simp} inference with premise  $N \cup \{ C\}$ and conclusion $N \cup \{ D\}$
is \emph{suitable (for the GBT fragment)} if 
\begin{enumerate}[\rm(i)]
\item
  for every BG term $t$
  occurring in  $\unabstr(D)$ there is a BG term $s \in \unabstr(C)$ such
  that $s \succeq_{\mathrm{fin}} t$,  
\item
  every occurrence of a BG-sorted FG
  operator $\op f$ in $\unabstr(D)$ is of the form $\op f(t_1,\ldots,t_n) \approx t$
  where $t$ is a ground BG term, 
\item
  every BG term in $D$ is pure, and 
\item
  if every BG term in $\unabstr(C)$ is ground then
  every BG term in $\unabstr(D)$ is ground.
\end{enumerate}

We say the \IR{Simp} inference rule is \emph{suitable} if
every \IR{Simp} inference is.
\qed
\end{definition}
Expected simplification techniques like demodulation, subsumption deletion and
evaluation of BG subterms are all covered as suitable \IR{Simp} inferences. 
Also, evaluation of BG subterms is possible, because simplifications are not only decreasing
\wrt\ $\succ$ but \emph{additionally} also decreasing \wrt\ $\succeq_{\mathrm{fin}}$,
as expressed
in condition (i). Without it, \eg,\ the clause $\op P(1+1,0)$ would admit
infinitely many simplified versions $\op P(2,0)$, $\op P(2,0+0)$, $\op P(2,0+(0+0))$, etc.

The $\HSP_\Base$ inferences do in general not preserve the shape of the clauses in
$\unabstr(N^\pre)$; they do preserve a somewhat weaker property -- cleanness --  which is sufficient for
our purposes. 
\begin{definition}[Clean clause]
\label{def:clean-clause}
A weakly abstracted clause $C$ is \emph{clean} if 
\begin{enumerate}[\rm(i)]
\item
  every BG term in $C$ is pure, 
\item
  every BG term in $\unabstr(C)$ is ground, and 
\item
  every occurrence of a BG-sorted FG
  operator $\op f$ in $\unabstr(C)$ is in a positive literal of the form $\op f(t_1,\ldots,t_n) \approx t$
  where $t$ is a ground BG term. 
\end{enumerate}
\end{definition}
For example,  if $\op c$ is FG-sorted, then $\op P(\op f(\op c) + 1)$ is not clean, 
while $\op f(x) \approx 1+\alpha \vee \op P(x)$ is. A clause set is called \emph{clean} if every clause in $N$
is. Notice that $N^\pre$ is clean. 

\begin{lemma}
\label{lemma:HSP-cleanness-preserved}
Let $C_1,\ldots,C_n$ be  clean clauses.  Assume a $\HSP_\Base$
inference  with premises  $C_1,\ldots,C_n$ and conclusion $C$. 
Then the following holds:
\begin{enumerate}[\rm(1)]
\item
  $C$ is clean.
\item
  Every BG term occurring in
  $\unabstr(C)$ also occurs in some clause $\unabstr(C_1)$, \ldots, $\unabstr(C_n)$. 
\end{enumerate}
\end{lemma}

  \begin{proof}
    Let $C'$ be the conclusion of the inference before weak abstraction, i.e.,
    $C = \abstr(C')$.  Regarding (1), property (i) of cleanness for $C$, we are given that
    all BG terms in all premise clauses are pure. Unification does not introduce general
    variables, and hence every BG term in $C'$ is pure.  Weak abstraction never introduces
    ordinary variables unless the given clause has ordinary variables. Thus all BG terms
    in $C$ are pure as well.

    The remaining properties (ii) and (iii) of cleanness of $C$ and property (2) can be
    seen from inspection of the $\HSP_\Base$ inference rules. We show it for superposition
    inferences, the other rules are similar. We distinguish three cases.

    Case 1: the inference is a superposition inference into a subterm $u$ of $s_i$ of a
    literal $s \approx t$ of the right premise, where
    $s = \op f(s_1,\ldots,s_i[u],\ldots,s_n)$ and $\op f$ is a BG-sorted foreground operator. With
    properties (i) and (iii) of cleanness holding for the premise clauses it follows that
    $u$ must be FG-sorted. In the conclusion of the inference $u$ is replaced by a
    FG-sorted term $r\sigma$ where $\sigma$ is the pivotal substitution.  By
    Prop.~\ref{prop:restricted-substitutions} the substitution $\sigma$ is restricted. This
    means that for every BG variable $X$ occurring in $C_1$ or $C_2$, if $X\sigma \neq X$ then
    $X\sigma$ is an (abstraction) variable occurring in $C_1$ or in $C_2$, or $X\sigma$ is a domain
    element.  Because neither variables nor domain elements are abstracted out, it follows
    that $C'$ does not require further abstraction, i.e.,\ $C = C'$.

    By item (ii) of cleanness, every BG term in $\unabstr(C_1)$ and in $\unabstr(C_2)$ is
    ground. In other words, every BG variable occurring in $C_1$ or $C_2$ is replaced by a
    ground term by unabstraction. This holds in particular for $X\sigma$ if $X\sigma$ is a variable,
    as opposed to a domain element. It follows that every BG term in
    $\unabstr(C') = \unabstr(C)$ is ground and that every BG term in $\unabstr(C)$ also
    occurs in $\unabstr(C_1)$ or $\unabstr(C_2)$, i.e., property (2) of the lemma
    claim. Finally observe that property (iii) of cleanness holds trivially for $C$, which
    completes the proof of property (1).

    Case 2: the inference is a superposition inference into any other FG-sorted subterm.
    This is proved in essentially the same way as in case 1.

    Case 3: the inference is a superposition inference into the top position of the left
    side of $\op f(s_1,\ldots s_n) \approx t$, where $\op f$ is a BG-sorted foreground operator. From
    item (iii) of cleanness it follows this term is replaced by a ground BG term $t'$,
    which does not need abstraction. The remaining argumentation is analogous to case 1
    and is omitted.  \qed
  \end{proof}

Thanks to conditions (ii)--(iv) in Def.~\ref{def:suitable-simp}, suitable \IR{Simp} inferences preserves cleanness:
\begin{lemma}
\label{lemma:Simp-cleanness-preserved}
Let $N \cup \{C\}$ be a set of clean clauses.  If $N \cup \{D\}$ is obtained 
from $N \cup \{C\}$  by a suitable \IR{Simp} inference then $D$ is clean.
\end{lemma}
\begin{proof}
Suppose $N \cup \{D\}$ is obtained from $N \cup \{C\}$  by a suitable \IR{Simp} inference.
  We need to show properties (i)--(iii) of cleanness for $D$. That every BG term in
  $D$ is pure follows from Def.~\ref{def:suitable-simp}-(iii).
That every BG term in $\unabstr(D)$ is ground follows from
Def.~\ref{def:suitable-simp}-(iv) and cleanness of $C$.
Finally, property (iii) follows from Def.~\ref{def:suitable-simp}-(ii).
\qed
\end{proof}

With the above lemmas we can prove our main result:
\begin{theorem}
\label{thm:hsp-gbt-complete}%
The $\HSP$ calculus with a suitable \IR{Simp} inference rule is dynamically refutationally
complete for the ground BG-sorted term fragment.  
More precisely, let $N$ be a finite set of GBT clauses  and $\calD = (N_i)_{i\ge 0}$ a fair
$\HSP$ derivation such that reckless \IR{Define} is is applied only in a
pre-derivation $(N_0 = \abstr(N)),\ldots,(N_k = N^\pre)$, for 
some $k \ge 0$. Then the limit of $\calD$ contains $\Box$ if and only if $N$ is
$\BAlgs$-unsatisfiable.
\end{theorem}
Notice that Thm.~\ref{thm:hsp-gbt-complete} does not appeal to compactness of 
BG specifications.

\begin{proof}
Our goal is to apply Thm.~\ref{thm:hspbase-ref-compl-lsc-define} and its proof, in a
slightly modified way. For that, we first need to know 
that $N^\pre = \abstr(N')$ for some clause set $N'$ that is locally sufficiently
complete.

We are given that $N$ is a set of GBT clauses. Recall that weak abstraction (recursively)
extracts BG subterms by substituting fresh variables and adding
disequations. Unabstraction reverses this process (and possibly eliminates additional
disequations).  It follows that with $N$ being a set of GBT clauses, so is
$\unabstr(\abstr(N)) = \unabstr(N_0)$.
From Lemma~\ref{lemma:define-preserves-GBT} it follows that $\unabstr(N^\pre)$ is also a
GBT clause set. 

Now chose $N'$ as the clause set that is obtained
from $N^\pre$ by replacing every clause $C \in N^\pre$ such that  $\unabstr(C)$
is a definition by $\unabstr(C)$. 
By construction of definitions, unabstraction reverses weak abstraction of definitions.
It follows $N^\pre = \abstr(N')$. By definition of pre-derivations,  all BG-sorted FG
terms occurring in $\unabstr(N^\pre)$ 
occur in definitions. Hence, with $\unabstr(N^\pre)$ being a set
of GBT clauses so is $N'$. It follows easily that $N'$ is locally sufficiently
complete, as desired.

We cannot apply Thm.~\ref{thm:hspbase-ref-compl-lsc-define} directly now because
it requires compactness of the BG specification, which cannot be assumed. However, we can
use the following argumentation instead.

Let $N^\infty = \bigcup_{i\ge 0} N_i$ be the union of $\calD$. 
We next show that
$\unabstr(N^\infty)$ contains only finitely many different BG terms and each of them is
ground. Recall that $\unabstr(N^\pre)$ is a GBT clause set, and so every BG term in 
$\unabstr(N^\pre)$ is ground. Because \IR{Define} is disabled in $\calD$, only
$\HSP_\Base$ and (suitable) \IR{Simp} inferences need to be analysed. Notice that
$N^\pre$ is clean and both the $\HSP_\Base$ and \IR{Simp} inferences preserve
cleanness, as per Lemmas~\ref{lemma:HSP-cleanness-preserved}-(1)
\looseness=-1
and~\ref{lemma:Simp-cleanness-preserved}, respectively.  

With respect to $\HSP_\Base$ inferences, together with
Def.~\ref{def:clean-clause}-(ii) it follows that every BG term $t$ in the
unabstracted version 
$\unabstr(C)$ of the inference conclusion $C$ is ground. Moreover,
$t$ also occurs in the unabstracted version of some premise clause by
Lemma~\ref{lemma:HSP-cleanness-preserved}-(2). In other words, $\HSP_\Base$ 
inferences do not grow the set of BG terms \wrt\ unabstracted premises and conclusions

With respect to \IR{Simp} inferences, 
$\unabstr(N^\pre)$ provide an upper bound \wrt\ the term ordering $\succ_{\mathrm{fin}}$ for all BG
terms generated in \IR{Simp} inferences. There can be only finitely many such terms,
and each of them is ground, which follows from
Def.~\ref{def:suitable-simp}-(i). 

Because every BG term occurring in $\unabstr(N^\infty)$ is ground, every BG clause in
$\unabstr(N^\infty)$ is a multiset of literals of the form $s \approx t$ or $s \not\approx t$, where $s$
and $t$ are ground BG terms. With only finitely many BG terms available, there are
only finitely many BG clauses in $\unabstr(N^\infty)$, modulo equivalence. Because
unabstraction is an equivalence transformation, there are only finitely many BG
clauses in $N^\infty$ as well, modulo equivalence. 

Let $N_\infty = \bigcup_{i\ge 0} \bigcap_{j\ge i} N_j$ be the limit clause set of the derivation $\calD$, which is saturated
\wrt\ the hierarchic superposition calculus and $\RedSgi$. Because $\calD$ is not a
refutation, it does not contain $\Box$. Consequently the \IR{Close} rule is not applicable to $N_\infty$.
The set $N^\infty$, and hence also $N_\infty \subseteq N^\infty$, contains only finitely many BG clauses,
modulo equivalence. This entails that the set of all $\Sigma_\back$-clauses in $\sgi(N_\infty)$
is satisfied by some term-generated $\Sigma_\back$-interpretation $I \in \BAlgs$.  Now, the
rest of the proof is literally the same as in the proof of
Thm.~\ref{thm:hspbase-ref-compl-lsc-define}. 
\qed
\end{proof}

Because unabstraction can also be applied to fully abstracted clauses, it is possible to
equip the hierarchic superposition calculus
of~\cite{Bachmair:Ganzinger:Waldmann:TheoremProvingHierarchical:AAECC:94}
with a correspondingly modified \IR{Define} rule and get Thm.~\ref{thm:hsp-gbt-complete}
in that context as well.

Kruglov and Weidenbach~\cite{Kruglov:Weidenbach:MACIS:2012}
have shown how to use hierarchic
superposition as a decision procedure for ground clause sets (and for Horn
clause sets with constants and variables as the only FG terms). Their method
preprocesses the given clause set by
``basification'', a process that removes BG-sorted FG terms similarly
to our reckless \IR{Define} rule. The resulting clauses then are fully abstracted and 
hierarchic superposition is applied. Some modifications of the inference rules
make sure derivations always terminate. Simplification is restricted to subsumption
deletion. The fragment of~\cite{Kruglov:Weidenbach:MACIS:2012} is a further restriction of the GBT fragment.
\looseness=-1
We expect we can get decidability results for that fragment with similar techniques.

\section{Linear Arithmetic}
\label{sec:linear-arithmetic}
For the special cases of linear integer arithmetic (LIA) and linear rational arithmetic
as BG specifications, the result of the previous section can be extended
significantly: In addition to ground BG-sorted terms, we can also permit
BG-sorted variables and, in certain positions, even variables with offsets.

Recall that we have assumed that equality is the only predicate
symbol in our language, so that a non-equational atom, say
$s < t$, is to be taken as a shorthand for the equation $(s < t) \approx \mathit{true}$.
We refer to the terms that result from this encoding of atoms
as \emph{atom terms};
other terms are called \emph{proper terms}.

\begin{theorem}
\label{thm:fo-vs-int}%
Let $N$ be a set of clauses over the signature of linear integer arithmetic
(with parameters $\alpha$, $\beta$, etc.),
such that
every proper term in these clauses is either
(i)~ground, or (ii)~a variable, or (iii)~a sum $\zeta + k$
of a variable $\zeta $ and a number $k \geq 0$ that occurs on the
right-hand side of a positive literal $s < \zeta + k$.
If the set of ground terms occurring in $N$ is finite,
then $N$ is satisfiable in LIA over $\mathbb{Z}$ if and only if
$N$ is satisfiable \wrt~the first-order theory of LIA.
\end{theorem}

\begin{proof}
Let $N$ be a set of clauses with the required properties,
and let $T$ be the finite set of ground terms occurring in $N$.
We will show that $N$ is equivalent to some \emph{finite}
set of clauses over the signature of linear integer arithmetic,
which implies that it is satisfiable
in the integer numbers if and only
if it is satisfiable in the first-order theory of LIA.

In a first step,
we replace every negative ordering literal $\neg s < t$ or $\neg s \leq t$
by the equivalent positive ordering literal $t \leq s$ or $t < s$.
All literals of clauses in the resulting set $N_0$ have the form
$s \approx t$, $s \not\approx t$, $s < t$, $s \leq t$,
or $s < \zeta + k$, where $s$ and $t$ are either variables or elements of $T$
and $k \in \mathbb{N}$.
Note that the number of variables in clauses in $N_0$ may be unbounded.

In order to handle the various inequality literals in a more
uniform way, we introduce new binary relation symbols
$<_k$ (for $k \in \mathbb{N}$) that are defined by $a <_k b$ if and
only if $a < b + k$.
Observe that $s <_k t$ entails $s <_n t$ whenever $k \leq n$.
Obviously, we may replace
every literal $s < t$ by $s <_0 t$,
every literal $s \leq t$ by $s <_1 t$,
and every literal $s < \zeta + k$ by $s <_k \zeta $.
Let $N_1$ be the resulting clause set.

We will now transform $N_1$ into an equivalent
set $N_2$ of ground clauses.
We start by eliminating all equality literals
that contain variables
by exhaustively applying the following transformation rules:
\[\begin{array}{@{}l@{\quad}c@{\quad}l@{\qquad}l@{}}
  N \cup \{\,C \lor \zeta \not\approx \zeta \,\}
& \to
& N \cup \{\,C\,\}
\\[0.5ex]
  N \cup \{\,C \lor \zeta \not\approx t\,\}
& \to
& N \cup \{\,C[\zeta \mapsto t]\,\}
& \text{if~} t \neq \zeta 
\\[0.5ex]
  N \cup \{\,C \lor \zeta \approx \zeta \,\}
& \to
& N
\\[0.5ex]
  N \cup \{\,C \lor \zeta \approx t\,\}
& \to
& N \cup \{\,C \lor \zeta <_1 t,\, C \lor t <_1 \zeta \,\}
& \text{if~} t \neq \zeta 
\end{array}\]
All variables in inequality literals
are then eliminated in a Fourier-Motzkin-like
manner by exhaustively applying the transformation rule
\[\begin{array}{@{}l@{\quad}c@{\quad}l@{}}
  N \cup \{\,C \lor \bigvee\limits_{i \in I} \zeta <_{k_i} s_i
               \lor \bigvee\limits_{j \in J} t_j <_{n_j} \zeta \,\}
& \to
& N \cup \{\,C \lor \bigvee\limits_{i \in I}\bigvee\limits_{j \in J} t_j <_{k_i+n_j} s_i\,\}
\end{array}\]
where $\zeta $ does not occur in $C$ and
one of the index sets $I$ and $J$ may be empty.

The clauses in $N_2$ are constructed over the finite set $T$ of
proper ground terms,
but the length of the clauses in $N_2$ is potentially unbounded.
In the next step,
we will transform the clauses in such a way that any pair of terms
$s,t$ from $T$ is related by at most one literal in any clause:
We apply one of the following transformation rules
as long as two terms $s$ and $t$ occur in more than one literal:
\[\begin{array}{@{}l@{\quad}c@{\quad}l@{\qquad}l@{}}
  N \cup \{\,C \lor s <_k t \lor s \approx t\,\}
& \to
& N \cup \{\,C \lor s <_k t\,\}
& \text{if~} k \geq 1
\\[0.5ex]
  N \cup \{\,C \lor s <_0 t \lor s \approx t\,\}
& \to
& N \cup \{\,C \lor s <_1 t\,\}
\\[0.5ex]
  N \cup \{\,C \lor s <_k t \lor s \not\approx t\,\}
& \to
& N
& \text{if~} k \geq 1
\\[0.5ex]
  N \cup \{\,C \lor s <_0 t \lor s \not\approx t\,\}
& \to
& N \cup \{\,C \lor s \not\approx t\,\}
\\[0.5ex]
  N \cup \{\,C \lor s <_k t \lor s <_n t\,\}
& \to
& N \cup \{\,C \lor s <_n t\,\}
& \text{if~} k \leq n
\\[0.5ex]
  N \cup \{\,C \lor s <_k t \lor t <_n s\,\}
& \to
& N
& \text{if~} k + n \geq 1
\\[0.5ex]
  N \cup \{\,C \lor s <_0 t \lor t <_0 s\,\}
& \to
& N \cup \{\,C \lor s \not\approx t\,\}
\\[0.5ex]
  N \cup \{\,C \lor L \lor L\,\}
& \to
& N \cup \{\,C \lor L\,\}
& \text{for any literal~} L
\\[0.5ex]
  N \cup \{\,C \lor s \approx t \lor s \not\approx t\,\}
& \to
& N
\end{array}\]
The length of the clauses in the resulting set $N_3$ is now bounded
by $\frac{1}{2}m(m+1)$, where $m$ is the cardinality of $T$.
Still, due to the indices of the relation symbols $<_k$,
$N_3$ may be infinite.
We introduce an equivalence relation $\sim$ on clauses in $N_3$
as follows:
Define $C \sim C'$ if for all $s,t \in T$
(i) $s \approx t \in C$ if and only if $s \approx t \in C'$,
(ii) $s \not\approx t \in C$ if and only if $s \not\approx t \in C'$,
and
(iii) $s <_k t \in C$ for some $k$
if and only if $s <_n t \in C'$ for some $n$.
This relation splits $N_3$ into at most
$(\frac{1}{2}m(m+1))^5$ equivalence classes.\footnote{%
  Any pair of terms $s,t$ is related in all clauses of an equivalence class
  by either a literal $s \approx t$, or $s \not\approx t$,
  or $s <_n t$ for some $n$,
  or $t <_n s$ for some $n$, or no literal at all,
  so there are five possibilities per unordered pair of terms.}

We will now show that each equivalence class 
is logically equivalent to a finite subset of itself.
Let $M$ be some equivalence class.
Since any two clauses from $M$
differ at most in the indices of their $<_k$-literals,
we can write every clause $C_i \in M$
in the form
\[\textstyle{
C_i ~~=~~ C \lor \bigvee\limits_{1 \leq l \leq n} s_l <_{k_{il}} t_l
}\]
where $C$ and the $s_l$ and $t_l$ are the same for all clauses in $M$.
As we have mentioned above, 
$s_l <_{k_{il}} t_l$ entails $s_l <_{k_{jl}} t_l$ whenever $k_{il} \leq k_{jl}$;
so a clause $C_i \in M$ entails $C_j \in M$
whenever the $n$-tuple
$(k_{i1},\dots,k_{in})$
is pointwise smaller or equal to the $n$-tuple
$(k_{j1},\dots,k_{jn})$
(that is, $k_{il} \leq k_{jl}$ for all $1 \leq l \leq n$).

Let $Q$ be the set of $n$-tuples of natural numbers corresponding
to the clauses in $M$.
By Dickson's lemma~\cite{Dickson1913},
for every set of tuples in $\mathbb{N}^n$
the subset of minimal tuples (\wrt~the pointwise
extension of $\leq$ to tuples) is finite.
Let $Q'$ be the subset of minimal tuples in $Q$,
and let $M'$ be the set of clauses in $M$ that correspond
to the tuples in $Q'$.
Since for every tuple in $Q \setminus Q'$ there is a smaller tuple in $Q'$,
we know that every clause in $M \setminus M'$ is entailed by some
clause in $M'$.
So the equivalence class $M$ is logically equivalent to its finite subset $M'$.
Since the number of equivalence classes is also finite
and all transformation rules are sound,
this proves our claim.
\qed
\end{proof}

In order to apply this theorem to hierarchic superposition,
we must again impose some restrictions on the calculus.
Most important, we have to change the definition of weak abstraction
slightly:
We drop the requirement that target terms
are not domain elements
from Def.~\ref{def:weak-abs},
\ie, we abstract out a non-variable BG term $q$
occurring in a clause  $C[f(s_1,\dots,q,\dots,s_n)]$,
where $f$ is a FG operator or at least one of the $s_i$
is a FG or impure BG term,
even if $q$ is a domain element.
As we mentioned already in Sect.~\ref{sec:weak-abstraction},
all results obtained so far hold also for the
modified definition of weak abstraction.
In addition, we must again restrict to \emph{suitable} \IR{Simp} inferences (Def.~\ref{def:suitable-simp}).
%
With these restrictions, we can prove our main result:

\begin{theorem}
The hierarchic superposition calculus is dynamically refutationally complete
\wrt~LIA over $\mathbb{Z}$
for finite sets of $\Sigma$-clauses
in which every proper BG-sorted term is either
(i)~ground, or (ii)~a variable, or (iii)~a sum $\zeta + k$
of a variable $\zeta $ and a number $k \geq 0$ that occurs on the
right-hand side of a positive literal $s < \zeta + k$.
\end{theorem}

\begin{proof}
Let $N$ be a finite set of $\Sigma$-clauses with the required properties.
By Lemma~\ref{lemma:prederiv-yields-lsc},
a pre-derivation starting from $N_0 = \abstr(N)$
yields a locally sufficiently complete
finite set $N_0$ of abstracted clauses.

Now we run the hierarchic superposition calculus on $N_0$
(with the same restrictions on simplifications as
in Sect.~\ref{sec:gbt}).
Let $N_1$ be the
(possibly infinite) set of BG clauses generated during the run.
By unabstracting these clauses, we obtain an equivalent set $N_2$ of
clauses that satisfy the conditions of Thm.~\ref{thm:fo-vs-int},
so $N_2$
is satisfiable in LIA over $\mathbb{Z}$ if and only if
$N$ is satisfiable \wrt~the first-order theory of LIA.
Since the hierarchic superposition calculus is dynamically refutationally complete
\wrt~the first-order theory of LIA, the result follows.
\qed
\end{proof}

Analogous results hold for linear rational arithmetic.
Let $n$ be the least common divisor of all numerical
constants in the original clause set; then we define
$a <_{2i} b$ by $a < b + \frac{i}{n}$
and
$a <_{2i+1} b$ by $a \leq b + \frac{i}{n}$
for $i \in \mathbb{N}$
and express every inequation literal in terms of $<_k$.
The Fourier-Motzkin transformation rule is replaced by
\[\begin{array}{@{}l@{\quad}c@{\quad}l@{}}
  N \cup \{\,C \lor \bigvee\limits_{i \in I} \zeta <_{k_i} s_i
               \lor \bigvee\limits_{j \in J} t_j <_{n_j} \zeta \,\}
& \to
& N \cup \{\,C \lor \bigvee\limits_{i \in I}\bigvee\limits_{j \in J} t_j <_{k_i \bullet n_j} s_i\,\}
\end{array}\]
where $\zeta $ does not occur in $C$,
one of the index sets $I$ and $J$ may be empty,
and $k \bullet n$ is defined as $k + n - 1$ if both $k$ and $n$ are odd,
and $k + n$ otherwise.
The rest of the proof proceeds in the same way as before.

The restriction to sums $\zeta + k$ occurring on \emph{right-hand} sides
of positive literals $s < \zeta + k$ is crucial.
Consider the following example:
Suppose that we have a unary FG predicate symbol $\op P$
and the $\Sigma$-clause set $N$
\begin{align*}
& \op P(0), \\[-0.4ex]
& \neg \op P(X) \:\lor\: X < \alpha, \\[-0.4ex]
& \neg \op P(X) \:\lor\: X + 1 < Y \:\lor\: \op P(Y)
\end{align*}
over linear integer arithmetic.
Note that $X + 1$ occurs on the left-hand side of the literal.

Starting with these clauses, hierarchic superposition produces
the following set $N_1$ of BG clauses
\begin{align*}
& 0 < \alpha, \\[-0.4ex]
& 0 + 1 < Y_1 \:\lor\: Y_1 < \alpha, \\[-0.4ex]
& 0 + 1 < Y_1 \:\lor\: Y_1 + 1 < Y_2 \:\lor\: Y_2 < \alpha, \\[-0.4ex]
& 0 + 1 < Y_1 \:\lor\: Y_1 + 1 < Y_2 \:\lor\: Y_2 + 1 < Y_3 \:\lor\: Y_3 < \alpha, \\[-0.7ex]
& \ldots
\end{align*}
After removing the universally quantified variables by quantifier
elimination, $N_1$ turns out to be equivalent
to $\{{0 < \alpha,}$ ${1 < \alpha,}$ ${2 < \alpha,}$ ${3 < \alpha,}$ $\ldots\}$.
Each finite subset of $N_1$
is satisfiable in $\mathbb{Z}$, and hence
in the first-order theory of LIA.
By compactness of first-order logic, $N_1$ itself is also
satisfiable in the first-order theory of LIA,
for instance in the non-standard model $\mathbb{Q} \times \mathbb{Z}$
with $0 := (0,0)$, $1 := (0,1)$, $\alpha := (1,0)$,
${(x,y) + (x',y') := (x+x',y+y')}$,
and a lexicographic ordering.
On the other hand
$N_1$, and hence $N$, is clearly unsatisfiable over $\mathbb{Z}$,
but hierarchic superposition is unable to detect this.

\section{Experiments}
\label{sec:experiments}
We implemented the $\HSP$ calculus in the theorem prover
\beagle.\footnote{\beagle is available at
  \url{https://bitbucket.org/peba123/beagle}.
  The distribution includes the (Scala) source code and a ready-to-run Java jar-file.}
\beagle is a testbed for rapidly trying out theoretical ideas but it is not
a high-performance prover (in particular it lacks indexing of any form).
The perhaps most significant calculus feature not yet
implemented is the improvement for linear 
integer and rational arithmetic of Sect.~\ref{sec:linear-arithmetic}.

\beagle's proof procedure and background reasoning, in particular for linear integer
arithmetic, and experimental results have been described
in~\cite{Baumgartner:Bax:Waldmann:Beagle:CADE:2015}. Here we only provide an
update on the experiments and report on complementary aspects not discussed
in~\cite{Baumgartner:Bax:Waldmann:Beagle:CADE:2015}. More specifically, our new
experiments are based on a more recent version of the TPTP problem
library~\cite{Sutcliffe:TPTP:2017} (by four years), and we discuss in more detail the
impact of the various calculus variants introduced in this paper. We also compare
\beagle's performance to that of other provers.

We tested \beagle on the first-order problems from the TPTP
library, version 7.2.0,\footnote{\url{http://tptp.org}} that involve some form of arithmetic, including non-linear,
rational and real arithmetics.  The problems in the TPTP are organized in categories, and
the results for some of them are quickly dealt with: none of the HWV-problems in the TPTP library
was solvable within the time limit and we ignore these below. We ignore also
  the SYN category as its sole problem is merely a syntax test, and the GEG category as
  all problems are zero-rated and easily solved by \beagle.

The experiments were run on a MacBook Pro with a 2.3 GHz Intel i7 processor and 16 GB
main memory.  The CPU time limit was 120 seconds (a higher time limit does not help
much solving more problems). Tables~\ref{tab:results-theorem} and \ref{tab:results-rating}
summarize the results for the problems with a known ``theorem'' or ``unsatisfiable''
status with non-zero rating.
\beagle can also solve some satisfiable problems,
  but most of them are rather easy and can be solved by the BG solver
  alone. Unfortunately, the TPTP does not contain reasonably difficult satisfiable
  problems from the GBT-fragment, which would be interesting for exploiting the
  completeness result of Sect.~\ref{sec:gbt}.

\newcommand{\mlbox}[2][t]{
{\renewcommand{\arraystretch}{1.0}%
\begin{tabular}[#1]{@{}r@{}}#2\end{tabular}}}

\begin{table}[htb]
\caption{Number of TPTP version 7.2.0 problems solved, of all non-zero rated ``theorem'' or ``unsatisfiable''
  problems involving any form of arithmetic. The flag settings giving the best result are
  in typeset in bold. The CPU time limit was 120 seconds. The column ``Any'' is the number of problems solved in the union of the
  four setting to its left. For the ``Auto'' column see the description of auto-mode in the
  main text further below. For auto-mode only, the CPU time limit was increased to 300 seconds.
}
\label{tab:results-theorem}
  \centering
  \begin{tabular}{lr@{\qquad}rrrrrr}
 & & \multicolumn{2}{c}{Ordinary variables} & \multicolumn{2}{c}{Abstraction variables} \\[1ex]
    Category   & \#Problems  & \mlbox[b]{BG simp\\  cautious} 
 & \mlbox[b]{BG simp \\  aggressive} & \mlbox[b]{BG simp \\ cautious} & \mlbox[b]{BG
                                                                        simp \\
    aggressive} & \quad Any & \quad Auto \\[1ex]\hline
\rule{0pt}{3ex}%
ARI & 
444 &  
356  &  
\bfseries 357 &
353  &  
355  &
362  &
355 \\
DAT & 
23 &  
9  &  
\bfseries 12 &
 6  &  
 7 &
 13 &
12
    \\
MSC & 
3 &  
3   &  
3 &  
 3 &  
3 &
 3 &
3 \\
NUM & 
36 &  
 30 &  
 29 &
\bfseries  34 &  
 \bfseries 34 &
 34 &
34 \\
PUZ & 
1 &  
1 &  
1 &
  1 &  
 1 &
 1 &
1 \\
SEV & 
2 &  
0  &  
0 &
0   &  
0  &
0  &
0 \\
SWV & 
1 &  
1 &  
1 & 
1  &  
1 &
1 &
1 \\
SWW & 
244 &  
91 &  
88 & 
\bfseries 92 &  
 89 &
 97 &
95 \\
SYO & 
1 &  
  0 &  
 0 &  
 0 &  
0 &
 0 &
0 \\[0.5ex]
\itshape Total & 755 & 419 & 471 &  \bfseries 490 &  \bfseries 490 & 511 &  501 \\\hline
 \end{tabular}
\end{table}
Table~\ref{tab:results-theorem} is a breakdown of \beagle's performance by TPTP
problem categories and four flag settings. \beagle features a host of flags for
controlling its search, but in Table~\ref{tab:results-theorem}  we varied only the two most influential
ones: one that controls whether input arithmetic variables are taken as ordinary
variables or as abstraction variables. (Sect.~\ref{sec:weak-abstraction} discusses the
trade-off between these two kinds of variables.) The other controls whether 
simplification of BG terms is done cautiously or aggressively.

To explain, the cautious simplification rules comprise evaluation of arithmetic terms,
\eg~$3\cdot 5$, $3 < 5$, $\alpha+1 < \alpha+1$ (equal lhs and rhs terms in inequations), and rules
for TPTP-operators, \eg, $\op{to\_rat}(5)$, $\op{is\_int}(3.5)$. For aggressive
simplification, integer sorted subterms are brought into a polynomial-like form and
are evaluated as much as possible.
For example, the term $5\cdot\alpha + \op f(3+6, \alpha\cdot 4) - \alpha\cdot 3$ becomes $2\cdot\alpha + \op
f(9, 4\cdot\alpha)$.  These conversions exploit the associativity and
commutativity laws for $+$ and $\cdot$. We refer the reader to~\cite{Baumgartner:Bax:Waldmann:Beagle:CADE:2015}
for additional aggressive simplification rules, but we note here that
aggressive simplification does not always preserve sufficient completeness. For
example, in the clause set $N = \{ \op P(1+(2+\op f(X))),\ \neg \op P(1+(X+\op f(X)))\}$
the first clause is aggressively simplified, giving $N' = \{ \op P(3+\op f(X)),\ \neg \op P(1+(X+\op f(X)))\}$. 
Both $N$ and $N'$ are LIA-unsatisfiable,
$\sgi(N) \cup \GndTh(\mathrm{LIA})$ is unsatisfiable, but 
$\sgi(N') \cup \GndTh(\mathrm{LIA})$ is satisfiable. Thus, $N$ is (trivially) sufficiently
complete while $N'$ is not. 

These two flag settings, in four combinations in total, span a range from
``most complete but larger search space'' by using
ordinary variables and cautious simplification, to ``most incomplete but smaller
search space'' by using abstraction variables and aggressive simplification.
As the results in Table~\ref{tab:results-theorem} show, the flag setting ``abstraction
variables'' solves more problems than ``ordinary variables'', but not uniformly so.
Indeed, as indicated by the ``Any'' column in Table~\ref{tab:results-theorem}, there are
problems that are solved only with either ordinary or abstraction variables.

Some more specific comments, by problem categories:

\paragraph{ARI} Of the 362 solved problems, 14 are not solved in every setting. Of these,
four problems require cautious simplification, and five problems require aggressive
simplification. This is independent from whether abstraction or ordinary variables are
used.

\paragraph{DAT} The DAT category benefits significantly from using ordinary
variables. There is only one problem, DAT075=1.p, that is 
not solved with ordinary variables.
Two problems, DAT072=1.p and DAT086=1.p are solvable only with ordinary
variables and aggressive simplification. 

Many problems in the DAT category, including DAT086=1.p, state \emph{existentially
  quantified} theorems about data structures such as arrays and lists. If they are of
an arithmetic sort, these existentially quantified variables must be taken as
ordinary variables. This way, they
can be unified with BG-sorted FG terms such as  
$\op{head}(\op{cons}(x, y))$ (which appear in the list axioms)  
which might be necessary for getting a refutation at all. 

A trivial example for this phenomenon is the entailment $\{\op P(\op f(1))\} \models \exists x\
\op P(x)$, where $\op f$ is BG-sorted, which is provable only with ordinary variables.

\paragraph{NUM} This category requires abstraction variables. With it, four
of the problems can be solved in the NUM category (NUM859=1.p, NUM860=1.p, NUM861=1.p,
NUM862=1.p), as the search space with ordinary variables is too big.

\paragraph{SWW} By and large, cautious BG simplification fares slightly better on the SWW
problems. Of the 97 problems solved, 16 are not solved in every
setting, and the settings that do solve it do not follow an obvious pattern. 

\newcommand{\circnr}[1]{\FPeval\result{clip(171+#1)}\ding{\result}}

We were also interested in \beagle's performance, on the same problems, broken
down by the calculus features introduced in this paper.
Table~\ref{tab:results-rating}
summarizes our findings for five configurations \circnr{1}-\circnr{5} obtained by
progressively enabling these features. In order to assess the usefulness of the features we filtered the
results by problem rating. The column ``$\geq 0.75$'', for instance, lists the number of
solved problems, of all 80 known ``theorem'' or ``unsatisfiable'' problems with a rating 0.75 or higher
and that involve some form of arithmetic.
\begin{table}[htb]
\caption{Number of ``theorem'' or ``unsatisfiable'' problems solved, by calculus features and problem rating, excluding the HWV-problems.}
\label{tab:results-rating}
  \centering
  \begin{tabular}{c@{\qquad}lr@{\qquad}r@{\qquad}r@{\qquad}r@{\qquad}r}
& & & \multicolumn{4}{c}{Rating, \# Problems}\\\cline{4-7}
& \rule{0pt}{3ex} &  &  $\ge 0.1$ & $\ge 0.5$ & $\ge 0.75$ & $\ge 0.88$ \\
& Abstraction & Feature & 756 & 187 & 80 & 55 \\[1ex]\hline
\circnr{1} & \rule{0pt}{3ex} Standard &  N/A & 355 & 30 & 5 & 1 \\
\circnr{2} &              & + \IR{Define} & 493 & 38 & 5 & 1 \\[1ex]\hline
\circnr{3} & \rule{0pt}{3ex} Weak & + \IR{Define} & 490 & 40& 5 & 1 \\[1ex]
\circnr{4} & \rule{0pt}{3ex}  & \mlbox[b]{+ \IR{Define} \\ + Ordinary vars} & 500 & 44 & 5 & 1  \\[1ex]
\circnr{5} & \rule{0pt}{3ex}               & \mlbox[b]{+ \IR{Define} \\ + Ordinary vars \\
    + BG  simp aggressive } & 511 & 45 & 5 & 1 \\
\hline
  \end{tabular}
\end{table}

The predecessor calculus
of~\cite{Bachmair:Ganzinger:Waldmann:TheoremProvingHierarchical:AAECC:94} uses an 
exhaustive abstraction mechanism that turns every side of an equation into either a pure
BG or pure FG term. All BG variables are always abstraction variables.
Configuration \circnr{1} implements this calculus, with the only deviation of an added 
splitting rule. The splitting rule~\cite{DBLP:conf/cade/WeidenbachDFKSW09}  breaks apart a
clause into variable-disjoint parts and 
leads to a branching search space for finding corresponding sub-proofs. See
again~\cite{Baumgartner:Bax:Waldmann:Beagle:CADE:2015} for more details.

In our experiments splitting is always enabled, in particular also for configuration
\circnr{1} for better comparability of result.
Cautious BG simplification is enabled for configuration \circnr{1} and the subsequent
configurations \circnr{2}-\circnr{4}. 

Configuration \circnr{2} differs from configuration \circnr{1} only by an additional
\IR{Define} rule. (As said earlier, the \IR{Define} rule can be added without
problems to the previous calculus.)
By comparing the results for \circnr{1} and
\circnr{2} it becomes obvious that adding \IR{Define} improves performance
dramatically. This applies to the new calculus as well.  The \IR{Define} rule 
stands out and should always be enabled.

Configuration \circnr{3} replaces the standard abstraction mechanism
of~\cite{Bachmair:Ganzinger:Waldmann:TheoremProvingHierarchical:AAECC:94} by the new
weak abstraction mechanism of Sect.~\ref{sec:weak-abstraction}. Weak abstraction seems more
effective than standard abstraction for problems with a higher rating, but the data
set supporting this conclusion is very small.

There are five problems, all from the SWW category\footnote{SWW583=2.p, SWW594=2.p,
  SWW607=2.p, SWW626=2.p, SWW653=2.p and SWW657=2.p} that re solved \emph{only} with
configuration \circnr{2}, and there is one problem, SWW607=2.p, that is solved only by
configurations \circnr{1} and \circnr{2}.

There are four solvable problems with rating 0.75. These are ARI595=1.p -- ARI598=1.p,
which are ``simple'' problems involving a 
free predicate symbols over the integer background theory. The problem ARI595=1.p, for
instance, is to prove the validity of the formula
$(\forall\, z{\,:\,}\mathds{Z}\ a \le z \land z \le a+2 \rightarrow p(z)) \rightarrow \exists\, x{\,:\,}\mathds{Z}\ p(3\cdot x)$.\footnote{At the time of this writing, there are only four provers (including \beagle) registered
  with the TPTP web infrastructure that can solve these problems. Hence the rating 0.75.}
The calculus and implementation techniques needed for solving such problems are
rather different to those needed for solving combinatory problems involving trivial
arithmetics only, like, e.g., the HWV-problems.

Configuration \circnr{4} is the same as \circnr{3} except that it includes the
results for general variables instead of abstraction variables. 
Similarly, configuration \circnr{5} is the same as \circnr{4} except that it includes the
results for aggressive BG simplification. It is the union of all results in
Table~\ref{tab:results-theorem}. 

For comparison with other implemented theorem provers for first-order logic with
arithmetics, we ran \beagle on the problem set used in  the 2018 edition of the CADE ATP system competition
(CASC-J9).\footnote{\url{http://tptp.cs.miami.edu/~tptp/CASC/J9/}}.  The competing
systems were CVC4~\cite{barrett2011cvc4},
Princess~\cite{Ruemmer:SequentCalculusLIA:LPAR:2008}, and
two versions of Vampire~\cite{DBLP:conf/cav/KovacsV13}. 

In the competition, the systems were given 200 problems from the TPTP problem library, 125
problems over the 
integers as the background theory (TFI category), and 75 over the reals (TFE category).
The system that solves the most problems in the union of the TFI and TFE categories within
a CPU time limit of 300 sec wins. We applied \beagle in an ``auto'' mode, which
time-slices (at most) three parameter settings. These differ mainly in their use of
abstraction variables or ordinary variables, and the addition of certain arithmetic
lemmas. 

\begin{table}[htb]
\caption{CADE ATP system competition results 2018 and \beagle's performance on the
  same problem sets.}
\label{tab:CASC}
  \centering
  \begin{tabular}{l@{\quad}r@{\quad}r@{\quad}r@{\quad}r@{\quad}r@{\quad}r}
    & 
\mlbox{Vampire\\ \footnotesize 4.3} &
\mlbox{Vampire\\ \footnotesize 4.1} &
\mlbox{CVC4\\ \footnotesize 1.6pre} & 
\mlbox{Princess\\ \footnotesize 170717} & 
\mlbox{Beagle\\ \footnotesize 0.9.51} \\[1ex]\hline
\rule{0pt}{3ex}    \#Solved TFI (of 125) & 93 & 98 & 85 & 62 & 36 \\
\rule{0pt}{3ex}    \#Solved TFE (of 75) &  70 & 64  & 72  & 43  & 44 \\
    \rule{0pt}{3ex}    \#Solved TFA (of 200) &  163 & 162  & 157  & 105  & 70 \\
    \hline
  \end{tabular}
\end{table}
The results are summarized in Table~\ref{tab:CASC}. We note that \beagle was run on
different hardware but the same  timeout of 300 seconds. The
results are thus only indicative of \beagle's performance, but we do not expect
significantly different result  had it participated.
In the TFI category, of the 36 problems solved, 5 require the use of ordinary variables.
In the TFE category, 16 problems involve the ceiling or floor function, which is currently
not implemented, and hence cannot be attempted.

In general, many problems used in the competition
are rather large in size or search space and would require a more sophisticated implementation
of \beagle. 

\section{Conclusions}
\label{sec:conclusions}
The main theoretical contribution of this paper is an improved variant of the
hierarchic superposition calculus. One improvement
over its
predecessor~\cite{Bachmair:Ganzinger:Waldmann:TheoremProvingHierarchical:AAECC:94}  
is a different form of ``abstracted'' clauses, the clauses the calculus
works with internally. Because of that, a modified completeness proof is required.
We have argued informally for the
benefits over the old calculus
in~\cite{Bachmair:Ganzinger:Waldmann:TheoremProvingHierarchical:AAECC:94}. They
concern 
making the calculus ``more complete'' in practice. It is hard to quantify that exactly
in a general way, as completeness is impossible to achieve in presence of
background-sorted foreground function symbols (\eg, ``head'' of integer-sorted
lists). To compensate for that to some degree, we have reported on initial experiments with a
prototypical implementation on the TPTP problem library. 
These experiments clearly
indicate the benefits of our concepts, in particular the 
definition rule and the use of ordinary variables. There is no problem that is solved only
by the old calculus only. Certainly
more experimentation and an improved implementation is needed to also solve
bigger-sized problems with a larger combinatorial search space.

We have also obtained two new completeness results for certain clause logic fragments
that do not require compactness of the background specification, cf.\
Sect.~\ref{sec:gbt} and Sect.~\ref{sec:linear-arithmetic}. The former is loosely related 
to the decidability results in~\cite{Kruglov:Weidenbach:MACIS:2012}, as discussed in
Sect.~\ref{sec:define}. It is also loosely related to results in SMT-based theorem
proving. For instance, the method in~\cite{Ge:DeMoura:CompleteInstantiation:CAV:2009} 
deals with the case that variables appear only as arguments of, in our words,
foreground operators. It works by ground-instantiating all variables in order to being
able to use an SMT-solver for the quantifier-free fragment. Under certain conditions,
finite ground instantiation is possible and the method is complete, otherwise it is
complete only modulo compactness of the background theory (as expected).
Treating different fragments, the theoretical results are mutually non-subsuming with
ours. Yet, on the fragment they consider we could adopt their technique of finite ground
instantiation before applying Thm.~\ref{thm:hsp-gbt-complete} (when it applies).
However, according to Thm.~\ref{thm:hsp-gbt-complete} our calculus needs
instantiation of \emph{background-sorted variables only}, this way keeping reasoning
with foreground-sorted terms on the first-order level, as usual with superposition.

\bibliographystyle{myabbrv}
\bibliography{bibliography}

\end{document}